\documentclass[12pt,a4paper]{article} 
\RequirePackage{etex}
\usepackage[colorlinks]{hyperref}
\usepackage{amsmath}
\usepackage{amsthm}
\usepackage{amsfonts}
\usepackage{amssymb}
\usepackage{mathabx}
\usepackage{bussproofs}
\usepackage{lineno}

\usepackage{thmtools,thm-restate}

\usepackage{cite}
\usepackage{amsmath,amssymb,amsfonts}
\usepackage{algorithmic}
\usepackage{graphicx}
\usepackage{textcomp}
\usepackage{xcolor}
\def\BibTeX{{\rm B\kern-.05em{\sc i\kern-.025em b}\kern-.08em
    T\kern-.1667em\lower.7ex\hbox{E}\kern-.125emX}}

\usepackage{stmaryrd}
\usepackage{euscript} 
\usepackage{latexsym}
\usepackage{mathrsfs}
\usepackage[all]{xy} 
\usepackage{dsfont}
\usepackage{proof}

\usepackage{graphicx}
\usepackage{relsize}

\usepackage{url}
\usepackage{enumitem}
\usepackage{xcolor}
\usepackage{colortbl}
\usepackage{colonequals}
\usepackage[all]{xy}
\usepackage{tikz}
\usetikzlibrary{positioning,arrows,calc}

\usepackage{graphics} 
\usepackage{mathrsfs}
\usepackage{xcolor}
\theoremstyle{plain}

\usepackage{booktabs,array}
\newcolumntype{P}[1]{>{\raggedright\arraybackslash}p{\dimexpr#1\linewidth-2\tabcolsep}}
\theoremstyle{plain} 
\newtheorem{theorem}{Theorem}
\newtheorem{proposition}[theorem]{Proposition}

\newtheorem{corollary}[theorem]{Corollary}

\theoremstyle{definition}
\newtheorem{definition}[theorem]{Definition}
\newtheorem{example}[theorem]{Example}
\newtheorem{remark}[theorem]{Remark}

\def\IQC{\mathsf{IQC}}

\def\QFS{\mathsf{QFS}}
\def\LJ{\mathbf{LJ}}
\def\LK{\mathbf{LK}}

\newcommand{\calF}{\mathcal{F}}
\newcommand{\calM}{\mathcal{M}}
\newcommand{\calD}{\mathcal{D}}
\def\phi{\varphi}

\newcommand{\limp}{\rightarrow}

\usepackage{amssymb}
\usepackage{amsmath}
\usepackage{amsthm}
\usepackage{amsfonts}
\usepackage{amssymb}
\usepackage{tikz}

\usepackage{graphicx} 
\usepackage{booktabs} 
\usepackage{bussproofs}
\usepackage{mathabx}

\graphicspath{ {./images/} }

\def\IQC{\mathsf{IQC}}

\def\SW{\mathrm{SW}}
\def\CD{\mathrm{CD}}
\def\ED{\mathrm{ED}}

\def\WF{\mathrm{WF}}
\def\cWF{\mathrm{cWF}}
\def\FDS{\mathrm{FDS}}
\def\QFS{\mathsf{QFS}}

\usepackage{bussproofs}
\usepackage{amssymb} 
\usepackage{amsmath} 
\usepackage{stmaryrd}
\usepackage{euscript} 
\usepackage{latexsym}
\usepackage{mathrsfs}
\usepackage[all]{xy}
\usepackage{dsfont}
\usepackage{proof}

\newcommand{\ipc}{{\sf IPC}}

\newcommand{\iqc}{{\sf IQC}}

\newcommand{\cqc}{{\sf CQC}}

\newcommand{\forc}{\Vdash}

\newcommand{\calc}{{\cal C}}

\newcommand{\De}{\Delta}
\newcommand{\Ga}{\Gamma}

\newcommand{\imp}{\rightarrow} 
\newcommand{\ifff}{\leftrightarrow}

\newcommand{\en}{\wedge} 
\newcommand{\of}{\vee} 
\newcommand{\E}{\exists}
\newcommand{\A}{\forall}

\usepackage{graphicx}
\usepackage{bussproofs}

\newcommand{\LKpp}{{\bf LK}^{++}}

\usepackage{amsmath}
\usepackage{amsfonts}
\usepackage{amssymb}
\usepackage{bussproofs}
\usepackage{footmisc}
\usepackage{mathtools}
\usepackage{bussproofs}
\usepackage{relsize}
\usepackage{scalerel}
\usepackage{lineno}
\usepackage{etex}
\usepackage{xcolor}
\usepackage{tikz}
\usepackage{tkz-graph}

\usepackage{tikz}
\usetikzlibrary{arrows}
\usepackage{tkz-graph}
\usetikzlibrary{backgrounds}

\usepackage{graphicx,amsmath}
\usepackage{pgf,tikz}
\usetikzlibrary{positioning,calc,chains,fit,shapes,arrows}

\newcommand{\LJpp}{{\bf LJ}^{++}}

\newcommand{\SA}{{{S}_A}} 
\newcommand{\SP}{{{S}_P}} 

\begin{document}

\title{Skolemization In Intermediate Logics
}


\author{Matthias Baaz$^*$, Mariami Gamsakhurdia\footnote{TU Wien, Austria, baaz@logic.at, mariami.gamsakhurdia.it@gmail.com}, \\Rosalie Iemhoff\footnote{Utrecht University, The Netherlands, r.iemhoff@uu.nl}, Raheleh Jalali\footnote{University of Bath, The UK, rjk53@bath.ac.uk}}

\maketitle

\begin{abstract}
Skolemization, with Herbrand’s theorem, underpins automated theorem proving and various transformations in computer science and mathematics. Skolemization removes strong quantifiers by introducing new function symbols, enabling efficient proof search algorithms. We characterize intermediate first-order logics that admit standard (and Andrews) Skolemization. These are the logics that allow classical quantifier shift principles. For some logics not in this category, innovative forms of Skolem functions are developed that allow Skolemization. Moreover, we analyze predicate intuitionistic logic with quantifier shift axioms and demonstrate its Kripke frame-incompleteness. These findings may foster resolution-based theorem provers for non-classical logics.  This article is part of a larger project investigating Skolemization in non-classical logics.
\end{abstract}
\noindent \textbf{Keywords:}
Skolemization, intermediate logics, quantifier shifts, Kripke frames, sequent calculus.

\section{Introduction}
Skolemization, introduced by Skolem in his seminal 1920 paper \cite{Skolem}, is a fundamental method in mathematics and computer science, with deep connections to both predicate and propositional logic. By replacing strong quantifiers (positive universal or negative existential quantifiers) with fresh function symbols, Skolemization simplifies first-order logic statements, making it particularly valuable in automated theorem proving and resolution methods \cite{automated,automated2,automated3}. The resolution method, one of the most prominent techniques for classical automated theorem proving in first-order logic, relies on the elimination of strong quantifiers through the introduction of Skolem functions, the deletion of weak quantifiers (negative universal and positive existential quantifiers), and the transformation of formulas into clause form—a conjunction of disjunctions of literals \cite{baaz2011methods}. Together with Herbrand’s theorem, Skolemization bridges predicate and propositional logic, serving as a foundation for establishing the decidability of theories and for practical applications in computer science.

The most widely used form of Skolemization is structural Skolemization (Definition \ref{Def: St Skolem}), in which strong quantifiers are replaced by Skolem functions that depend on all weak quantifiers within whose scope the replaced quantifier occurs. Skolemization, together with Herbrand's theorem, forms the basis of many transformations in computer science, especially in automated theorem proving \cite[Chapters 5,6]{automated}. In mathematics, Skolemization is viewed as a transformation of axioms to guarantee the existence of term models, while in proof theory, it is treated as a transformation of formulas that preserves validity. From this proof-theoretic perspective, Skolemization replaces strong quantifiers with function symbols, yielding a validity-equivalent formula. This transformation can also be applied to proofs: given a proof $\pi$ of a formula $A$, one can construct a proof $\pi'$ of the Skolemized version of $A$ by instantiating free variables.

Although Skolemization is sound and complete for classical logic, its applicability diminishes in non-classical settings, such as intuitionistic predicate logic ($\mathsf{IQC}$). While Skolemization is sound for many intermediate logics, it is often incomplete, as in the case of $\mathsf{IQC}$, where many classical quantifier shift principles fail. Even for prenex sequents (formulas in prenex normal form), Skolemization is sound and complete in intuitionistic logic \cite{Mints66,Mints72}, but this result does not extend to all sequents, since not every sequent in intuitionistic logic has a prenex form. This raises the important question of which intermediate first-order logics admit Skolemization and under what conditions.

In this article, we characterize the intermediate first-order logics that admit standard Skolemization, showing that these are precisely the logics that allow all classical quantifier shift principles (Theorem \ref{mainth}). Specifically, we introduce a logic $\QFS$ that contains these principles and show that Skolemization is sound and complete for $L$ if and only if $L$ contains $\QFS$. To the best of our knowledge, this is the first such characterization of Skolemizable logics, despite the technique's long history dating back to 1920. We extend our results to Andrews Skolemization (introduced in \cite{Andrews1,Andrews2} and given in Definition \ref{Def: And Skolem}), a variant of standard Skolemization in which Skolem functions depend only on weak quantifiers of the scope which bind in the subsequent formula. We show, additionally, that $\QFS$ is not the end of the story for Skolemization. We provide an example of a first-order intermediate logic that admits Skolemization with a different type of Skolem function, even though it does not satisfy all classical quantifier shifts (Example \ref{SL}).

We rely on proof-theoretic methods to achieve this result because, as established in Theorem \ref{thm: incompleteness}, Kripke-frame completeness fails for $\QFS$, ruling out the usual semantic approaches.
These limitations in semantic frameworks necessitate a deeper exploration of Skolemization from a proof-theoretic standpoint.

This article is part of a larger research project on Skolemization in non-classical logics. The first aim of this project is to characterize first-order logics that admit standard and Andrews Skolemization. The second aim is to develop innovative forms of Skolem functions for logics where standard Skolemization is not applicable (Section \ref{sec: SL}). The outcomes of this research have the potential to advance resolution-based theorem provers for non-classical logics, contributing to the broader intersection of logic, computer science, and proof theory.

\section{Preliminaries}
We work with the first-order language $\mathcal{L}$, which consists of infinitely many variables, denoted by $x, y, \dots$, possibly with indices, the connectives $\wedge, \vee, \to$, constants $\bot, \top$, quantifiers $\forall, \exists$, infinitely many $n$-ary predicate and $n$-ary function symbols,  for any natural number $n \geq 0$. Predicate symbols with arity $0$ are also called \emph{propositional variables}. \emph{Terms} and \emph{formulas} are defined inductively as usual. We define $\neg A$ as $A \to \bot$. The occurrence of $x$ in $A$ is \emph{bound} if{f} it is in the scope of a quantifier; otherwise, it is called \emph{free}. A formula is \emph{closed} if it contains no free variables. 
Denote the Intuitionistic Propositional Logic by $\mathsf{IPC}$ and Intuitionistic Predicate Logic by $\mathsf{IQC}$; the corresponding classical systems are denoted $\mathsf{CPC}$ and $\mathsf{CQC}$.

\begin{definition}
Each (predicate) \emph{logic} $L$ is identified with the set of its consequences, 
which is a set of formulas closed under the axioms and rules of $\iqc$.  
We use $\phi \in L$ and $L \vdash \phi$ interchangeably. For a logic $L$ and a set of formulas $\{A_i\}_{i \in I}$, by $L+\{A_i\}_{i \in I}$ we mean the smallest logic containing $L$ and all the formulas in $\{A_i\}_{i \in I}$.  
 A predicate logic $L$ is called \emph{intermediate} when $\mathsf{IQC} \subseteq L \subseteq \mathsf{CQC}$.
 \end{definition}

\subsection{Sequent Calculi}

\begin{definition}
A \emph{Sequent} $S$ is an expression $\Gamma \Rightarrow \Delta$, where $\Gamma$, the \emph{antecedent}, and $\Delta$,  the \emph{succedent} of the sequent, are finite multisets of formulas. The intended meaning of $S$ is $\bigwedge \Gamma \to \bigvee \Delta$. By $G \vdash^{\pi} S$ we mean the sequent $S$ is provable in the calculus $G$ by the proof $\pi$ and call $S$ the \emph{end-sequent}. 
\end{definition}


\begin{table}[t]
 \centering
\footnotesize \begin{tabular}{c c}
\vspace{5pt} 

$p \Rightarrow p$ \; \scriptsize{(Ax)} & $\bot \Rightarrow $\; \scriptsize{($\bot$)}\\

\vspace{5pt}

\AxiomC{$\Gamma \Rightarrow \Delta$}
 \RightLabel{\scriptsize{$(L w)$} }
\UnaryInfC{$A, \Gamma \Rightarrow \Delta$}
\DisplayProof
&  
\AxiomC{$\Gamma \Rightarrow \Delta$}
 \RightLabel{\scriptsize{$(R w)$} }
\UnaryInfC{$\Gamma \Rightarrow \Delta, A$}
 \DisplayProof\\

\vspace{5pt}

\AxiomC{$A, A, \Gamma \Rightarrow \Delta$}
 \RightLabel{\scriptsize{$(L c)$} }
\UnaryInfC{$A, \Gamma \Rightarrow \Delta$}
\DisplayProof
&  
\AxiomC{$\Gamma \Rightarrow \Delta, A, A$}
 \RightLabel{\scriptsize{$(R c)$} }
\UnaryInfC{$\Gamma \Rightarrow \Delta, A$}
 \DisplayProof\\
 
\vspace{5pt}

\AxiomC{$A,\Gamma \Rightarrow \Delta$}
 \RightLabel{\scriptsize{$(L \wedge_1)$} }
\UnaryInfC{$A \wedge B, \Gamma \Rightarrow \Delta$}
\DisplayProof
&
\AxiomC{$B,\Gamma \Rightarrow \Delta$}
 \RightLabel{\scriptsize{$(L \wedge_2)$} }
\UnaryInfC{$A \wedge B, \Gamma \Rightarrow \Delta$}
\DisplayProof\\
 
\vspace{5pt}
\AxiomC{$\Gamma \Rightarrow \Delta, A$}
 \AxiomC{$\Gamma \Rightarrow \Delta, B$}
 \RightLabel{\scriptsize{$(R \wedge)$} }
\BinaryInfC{$\Gamma \Rightarrow \Delta, A \wedge B$}
 \DisplayProof
 &
\AxiomC{$\Gamma \Rightarrow \Delta, A$}
 \RightLabel{\scriptsize{$(R \vee_1)$} }
 \UnaryInfC{$\Gamma \Rightarrow \Delta, A \vee B$}
 \DisplayProof\\

\vspace{5pt}

\AxiomC{$\Gamma \Rightarrow \Delta,  B$}
 \RightLabel{\scriptsize{$(R \vee_2)$} }
 \UnaryInfC{$\Gamma \Rightarrow \Delta, A \vee B$}
 \DisplayProof
 &
 \hspace{-30pt}
 \AxiomC{$A, \Gamma \Rightarrow \Delta$}
 \AxiomC{$B, \Gamma \Rightarrow \Delta$}
 \RightLabel{\scriptsize{$(L \vee)$} }
\BinaryInfC{$A \vee B, \Gamma \Rightarrow \Delta$}
 \DisplayProof \\

 \vspace{5pt}

\AxiomC{$\Gamma \Rightarrow A, \Delta$}
 \AxiomC{$\Gamma, B \Rightarrow \Delta$}
 \RightLabel{\scriptsize{$(L \to)$} }
 \BinaryInfC{$\Gamma,  A \to B \Rightarrow \Delta$}
\DisplayProof &
\AxiomC{$\Gamma, A \Rightarrow B, \Delta$}
 \RightLabel{\scriptsize{$(R \to)$} }
 \UnaryInfC{$\Gamma \Rightarrow A \to B, \Delta$}
\DisplayProof \\

\vspace{5pt}

\AxiomC{$\Gamma, A(t) \Rightarrow \Delta$}
\RightLabel{\scriptsize$(\forall L)$} 
 \UnaryInfC{$\Gamma, \forall x A(x) \Rightarrow \Delta$} \DisplayProof & 
 \AxiomC{$\Gamma \Rightarrow A(y), \Delta$}
\RightLabel{\scriptsize$(\forall R)$} \UnaryInfC{$\Gamma \Rightarrow \forall x A(x), \Delta$} \DisplayProof\\

\vspace{5pt}

\AxiomC{$\Gamma, A(y) \Rightarrow \Delta$}
\RightLabel{\scriptsize$(\exists L)$} 
 \UnaryInfC{$\Gamma, \exists x A(x) \Rightarrow \Delta$} \DisplayProof & 
 \AxiomC{$\Gamma \Rightarrow A(t), \Delta$}
\RightLabel{\scriptsize$(\exists R)$}  \UnaryInfC{$\Gamma \Rightarrow \exists x A(x), \Delta$} \DisplayProof\\
 
\AxiomC{$\Gamma \Rightarrow \Delta, A$}
 \AxiomC{$A, \Gamma \Rightarrow \Delta$}
 \RightLabel{\scriptsize{$(cut)$} }
\BinaryInfC{$\Gamma \Rightarrow \Delta $} 
\DisplayProof 
\end{tabular}
    \caption{First-order $\LK$.} In (Ax), $p$ is an atom. In $(\exists L)$ and $(\forall R)$ the variable $y$ is not free in the conclusion, known as the \emph{usual eigenvariable conditions}
    \label{table: LK}
\end{table}

\normalsize The usual sequent calculus for $\cqc$, called first-order $\mathbf{LK}$, is defined in Table \ref{table: LK}, where the  weak quantifiers only apply to terms without bound variables. First-order $\mathbf{LJ}$, which is the usual sequent calculus for $\iqc$ is defined as the \emph{single-conclusion version} of first-order $\mathbf{LK}$, i.e., everywhere in the succedents, there is at most one formula. 

\subsection{Kripke Frames and Models} 
We define Kripke frames and models for predicate intuitionistic logic, see \cite[Chapter 14]{Mints}, \cite[Section 5.11]{TroVanDalen}, and \cite{RoseRobert}. The \emph{graph} of an $n$-ary function $f: D^n \to D$ is defined as graph$(f):=\{(\bar{e}, d) \in D^{n+1} \mid f(\bar{e})=d\}$.

\begin{definition} \label{Def: Kripke frame}
A \emph{Kripke frame} for $\mathsf{IQC}$ is a triple $(W, \preccurlyeq, D)$, where $W \neq \emptyset$ is a set of worlds, $\preccurlyeq$ is a partial order (i.e., a binary reflexive, transitive, and antisymmetric relation) on $W$, and $D$ is a function assigning to each $w \in W$ a non-empty $D_w$, called the \emph{domain} of $w$, such that if $w \preccurlyeq w'$ then  $D_w \subseteq D_{w'}.$ To every $w \in W$ we assign a language $\mathcal{L}(D_w)$ which is an extension of the language of $\mathsf{IQC}$, $\mathcal{L}$, by adding new constants for all the elements of $D_w$. 
\end{definition}

\begin{definition} \label{Def: Kripke model}
A \emph{Kripke model} for $\mathsf{IQC}$ is a tuple $\mathcal{M}=(W, \preccurlyeq, D, V)$, where $(W, \preccurlyeq, D)$ is a Kripke frame and $V$ is a valuation function such that:
 
\noindent $\bullet$ for any propositional atom $p$ we have $V(p) \subseteq W$ and $V$ is \emph{persistent}, i.e., if $v \in V(p)$ and $v \preccurlyeq w$, then $w \in V(p)$;
 
\noindent $\bullet$ for any $n$-ary relation symbol $R$ we have $V(R)=$ $\{R_v \mid v \in W\}$ such that $R_v \subseteq D_v^n$ and if $v \preccurlyeq w$ then $R_v \subseteq R_w$;
 
\noindent $\bullet$ for any $n$-ary function symbol $f$ we have $V(f)=$ $\{f_v \mid v \in W\}$ such that $f_v: D_v^n \to D_v$ is a function and if $v \preccurlyeq w$, then $\text{graph}(f_v) \subseteq$ $\text{graph}(f_w)$.

We define the \emph{forcing} relation $\Vdash$ in a model $\mathcal{M}$ at node $v$ by induction:
\vspace{-10pt}
\begin{center}
\begin{tabular}{l}
 $\calM, v \Vdash \top$ \qquad \;$\calM, v \not \Vdash \bot$ \\
  $\calM, v \Vdash p$ \; if{f} \; $v \in V(p)$, for a propositional atom $p$\\
 $\calM, v \Vdash A \wedge B$\; if{f} \; $\calM, v \Vdash A$ and $\calM, v \Vdash B$ \\
 $\calM, v \Vdash A \vee B $ \; if{f} \; $\calM, v \Vdash A$ or $\calM, v \Vdash B$\\
 $\calM, v \Vdash A \to B $  if{f}  $ \forall w \succcurlyeq v $ if $\calM, w \Vdash A$ then $\calM, w \Vdash B$\\
  $\calM, v \Vdash R(x_0, \ldots, x_{n-1})$  \; if{f} \; $(x_0, \ldots, x_{n-1}) \in R_v$,\\
  $\calM, v \Vdash f(x_0, \ldots, x_{n-1})=y$ \; if{f} \; $f_v(x_0, \ldots, x_{n-1})=y$\\
  $\calM, v \Vdash \exists x A(x)$  \; if{f} \; $\exists d \in D_v$ such that $\mathcal{M}, v \Vdash A(d)$\\
  $\calM, v \Vdash \forall x A(x)$  \; if{f} \;  $ \forall w \succcurlyeq v$, $\forall d \in D_w$ \; $\mathcal{M}, w \Vdash A(d)$
\end{tabular} 
\end{center}

A formula $A$ is defined to be \emph{valid} in a frame $F$, denoted by $F \vDash A$, and valid in a model $\calM$, denoted by $\calM \vDash A$, as usual. For any $w \in W$ define $\preccurlyeq [w]:=\{v \in W \mid v \succcurlyeq w\}$. 
\end{definition}

\begin{definition} \label{Def: linear}
A Kripke frame is \emph{linear} when 
\[
\forall u,v,w \in W \ \text{if} \ u \succcurlyeq w \ \text{and} \ v \succcurlyeq w \ \text{then either} \ u \succcurlyeq v \ \text{or} \ v \succcurlyeq u.
\]
A Kripke frame is \emph{constant domain} when
\[
\forall u,v \in W  \ \text{if} \ u \succcurlyeq v \ \text{then} \ D_u=D_v.
\]
\end{definition}
Observe the subtle distinction between our definition and the standard definitions of linear and constant domain frames. In our framework, disjoint unions of linear frames (as defined traditionally) are also considered linear$-$a feature that will play a key role in Theorem \ref{thm: incompleteness} and Section \ref{semantic-proof}.

\begin{remark}
A Kripke frame (model) for any logic stronger than $\IQC$ is a Kripke frame (model) for $\IQC$, as well. Therefore, as we are only considering intermediate logics in this paper, in any frame or model, we assume $\IQC$ is validated.    
\end{remark}

\begin{definition}
The frame $F$ is said to be \emph{finite} (resp. \emph{infinite}) when the set of worlds $W$ is finite (resp. infinite).
We say the worlds $u,v,w$ form a \emph{fork} in a frame when $w \preccurlyeq u$ and $w \preccurlyeq v$. 
Recall that a binary relation $\preccurlyeq$ on a set $X$ is called \emph{well-founded} if every non-empty subset $S \subseteq X$ has a minimal element w.r.t. $\preccurlyeq$, i.e., there is no infinite descending chain $\dots \preccurlyeq x_2 \preccurlyeq x_1 \in X$. Similarly, $\preccurlyeq$ on $X$ is called \emph{conversely well-founded} if every non-empty subset $S \subseteq X$ has a maximal element, i.e., $X$ does not contain an infinite ascending chain $x_1 \preccurlyeq x_2 \preccurlyeq \dots \in X$.
\end{definition}

\subsection{Skolemization}
\begin{definition}
An occurrence of a quantifier is \emph{positive} in a formula $A$ if it is in the scope of the antecedent of an even number of implication symbols. It is called \emph{negative} otherwise. 
\end{definition}
Note that the above definition covers negation as well, as it is defined via implication.
\begin{definition}
A quantifier is \emph{strong in a formula}
if it is a positive occurrence of a universal quantifier or a negative occurrence of an existential quantifier. It is called \emph{weak}, otherwise. 
A quantifier is \emph{strong in a sequent} $\Gamma \Rightarrow \Delta$ if it is a strong quantifier in a formula in $\Delta$ or a weak quantifier in a formula in $\Gamma$. Conversely, we can define a weak quantifier in a sequent. The \emph{first strong
quantifier} in a formula $A$ is the first strong quantifier in $A$ when reading $A$ from left to
right.
\end{definition}

There are different forms of Skolemization. Here we recall the standard one, which is also called structural Skolemization in the literature, introduced by Skolem in his seminal paper \cite{Skolem} and Andrews Skolemization \cite{Andrews1,Andrews2}. In Section \ref{sec: SL}, we will see another innovative form of Skolemization.

\begin{definition}\label{Def: St Skolem} 
 Let $A$ be a closed first-order formula. Whenever $A$ does not contain strong quantifiers, we define its \emph{structural Skolem form} as $A^S := A$. Suppose now that $A$ contains strong quantifiers. Let $Q \in \{\forall, \exists\}$ and $(Q y)$ be the first strong quantifier occurring in $A$. If $(Q y)$ is not in the scope of weak quantifiers, then its structural Skolem form is
$$ A^S := (A_{-(Q y)}\{y \leftarrow c\})^S,$$
where $A_{-(Q y)}$ is the formula $A$ after omission of $(Q y)$ and $c$ is a new constant symbol not occurring in $A$.
If $(Q y)$ is in the scope of the weak quantifiers $(Q_1 x_1) \ldots (Q_n x_n)$, then its structural Skolemization is
$$ A^S:= (A_{-(Q y)}\{y \leftarrow f(x_1, \ldots , x_n)\})^S,$$
where $f$ is a fresh (Skolem) function symbol and does not occur in $A$. 
Given $F = \bigwedge \Gamma \to \bigvee \Delta$ and $F^S= \bigwedge \Pi \to \bigvee \Lambda$, define the structural Skolemization of the sequent $\Gamma \Rightarrow \Delta$ as $$(\Gamma \Rightarrow \Delta)^S := \Pi \Rightarrow \Lambda.$$
\end{definition}

\begin{remark}
The Skolemization used in computer science, e.g., in resolution, is structural Skolemization: it is a Skolemization of a sequent $A \Rightarrow $. However, a proof in computer science is different from our proofs as the axiom of choice is used. The axiom of choice is somewhat stronger than necessary, but this proof verifies in addition the identity axioms of Skolem functions. The possibility of this argument is based on the existence of classical semantics.
\end{remark}

\begin{definition}\label{Def: And Skolem} 
For a closed first-order formula $A$, \emph{Andrews Skolemization} $A^\SA$ is defined similar to Definition \ref{Def: St Skolem}.
The only difference is that the introduced Skolem functions in Andrews method do not depend on the weak quantifiers $(Q_1 x_1) \ldots (Q_n x_n)$ that dominate the strong quantifier $(Q x)$, but depend on the subset of $\{x_1, \ldots, x_n\}$ appearing (free) in the subformula dominated by $(Q x)$.
\end{definition}
\begin{remark} 
It is worth mentioning that in \cite{BaazLolic} it is shown that under the assumption of any elementary clause form transformation, Andrews Skolemization might lead to a non-elementary speed-up compared to structural Skolemization, due to its property of leading to smaller Skolem terms. 
\end{remark}

\begin{definition}
Let $L$ be a logic.
We say Skolemization is \emph{sound and complete} for $L$, or $L$ \emph{admits} Skolemization, if for any formula $A$
\[
L \vdash A \quad \text{if and only if} \quad L \vdash A^S.
\]
The $(\Rightarrow)$ direction is sometimes called the \emph{Skolemization} and the $(\Leftarrow)$ direction is called the \emph{deSkolemization}. A similar definition holds for Andrews Skolemization.
\end{definition}
It is well-known that  Skolemization is sound and complete for $\cqc$.

\begin{example}\label{Sle} 
An example of Skolemization failing in $\IQC$ is the axiom $\SW: (\forall x A(x) \to B) \to \exists x (A(x) \to B)$, which is not provable in $\IQC$, as it can be refuted by a suitable model $\calM$. Let $\calM$ be constant domain, denote the domain by $\calD$, and $|\calD|>1$. Let $a,b \in \calD$ be distinct elements, $Q$ be a nullary and $P$ a unary predicate. Construct $\calM$ as follows, where $Q$ is not forced anywhere:

\small \[
  \xymatrix{
  v_1 \vDash P(a) \ v_1 \nVdash P(b) & &  
  v_2 \vDash P(b) \ v_1 \nVdash P(a) \\
  & w  \ar[lu] \ar[ru] & 
  } 
\]
\normalsize It is easy to see that
\[
\calM, w \Vdash \forall x P(x) \to Q \qquad \calM, w \nVdash \exists x (P(x) \to Q).
\]
However, the Skolemization of $\SW$ is valid in $\IQC$:
\[
\IQC \vdash (A(c) \to B) \to \exists x (A(x) \to B).
\]
\end{example}

\section{Logic of quantifier shifts}\label{semantic}
We introduce an interesting intermediate logic, which we call the \emph{logic of quantifier shifts}, denoted by $\QFS$. This logic plays a pivotal role in determining whether an intermediate logic $L$ can be Skolemized: Skolemization is sound and complete for $L$ if and only if $\QFS \subseteq L$. This result provides a compelling motivation for a deeper investigation of $\QFS$. 

\begin{definition}
Let $A(x)$ and $B$ be formulas in the language $\mathcal{L}$ and the variable $x$ is not free in $B$. The following formulas are called the \emph{quantifier shift} formulas:
\item 
 Constant Domain (CD):
    \[\forall x (A(x) \vee B) \to \forall x A(x) \vee B\]
    \item
   Existential Distribution (ED) 
   \[(B \to \exists x A(x)) \to \exists x (B \to A(x))\]
\item
    Quantifier Switch (SW):
    \[(\forall x A(x) \to B) \to \exists x (A(x) \to B)\]
    If we consider the commuting of quantifiers $\forall, \exists$ with the connectives $\wedge, \vee, \to$, then the quantifier shift formulas are the only ones not valid in $\IQC$.
However, in $\mathsf{CQC}$, both these formulas and their converses are provable  
Denote the logic $\mathsf{IQC} +\{\mathrm{CD}, \ED, \mathrm{SW}\}$ by $\mathsf{QFS}$. \\
\end{definition}

The following theorem shows the incompleteness of $\QFS$ and some fragments with respect to Kripke frames. 


\begin{restatable}{theorem}{ThmIncompleteness}
\label{thm: incompleteness}
The logics $\mathsf{QFS}$, $\mathsf{IQC}+ \{\mathrm{CD} , \mathrm{SW}\},$ and $\mathsf{IQC}+ \{\mathrm{CD} , \mathrm{ED}\}$ are all frame-incomplete.
\end{restatable}
We leave out the proof in this section. We do this in order to retain the focus on the proof-theoretic approach to proving our paper's core result.
We provide a precise proof and thorough semantic argument for the incompleteness of $\QFS$ and some fragments with respect to Kripke frames in Section \ref{semantic-proof}.


It is worth highlighting some historical studies in this area. In \cite[Section 2]{OnoProblems}, several logics related to $\QFS$ are discussed. Casari \cite{casari1987} provides results on frames for $\SW$ and $\ED$, while Takano \cite{takano1987} offers notable frame characterizations for certain logics by adding intuitionistically invalid formulas to $\IQC$. The frame incompleteness of $\IQC + \SW$ and $\IQC + \ED$ was demonstrated by Komori \cite{komori1983} and Nakamura \cite{nakamura1983}, respectively. A wide range of extensions of $\IQC + \{\mathrm{Lin}, \CD, \ED\}$ are explored in \cite{OnoCasari}, where $\mathrm{Lin}$ is defined as $(C \to D) \vee (D \to C)$ for any formulas $C$ and $D$. See also the excellent paper by Umezawa \cite{umezawa1959}, where he considers a rich class of logics resulted from $\IQC$ by addition of classically valid axiom schemes.

\section{Turning to Proof Theory}
In this section, we provide an elegant characterization of all intermediate logics that satisfy the soundness and completeness of Skolemization.
\begin{theorem}\label{mainth}
An intermediate logic admits structural or Andrews  Skolemization if and only if it contains all quantifier shift principles.
\end{theorem}
We now proceed to prove the above theorem, starting with one direction.
\subsection{From deSkolemization to quantifier shifts}

\begin{theorem}
Let $L$ be an intermediate logic. If Skolemization is sound and complete for $L$, then $\QFS \subseteq L$. 
\end{theorem}

\begin{proof}
Assume Skolemization is sound and complete for $L$, i.e., for any formula $\phi$
\[
L \vdash \phi \qquad \text{if{f}} \qquad L \vdash \phi^S.
\]
As $\iqc \subseteq L$, we only have to prove that $\{\CD, \ED, \SW\} \subseteq L$.

$\bullet$ The case $\CD \in L$. It is easy to see that the only strong quantifier in $\CD= \forall x (A(x) \vee B) \to \forall x A(x) \vee B$ is the universal quantifier in the succedent of the implication, which is not in the scope of any weak quantifiers. Therefore,
\[
\CD^S = \forall x (A(x) \vee B) \to A(c) \vee B,
\]
where $c$ is a fresh constant symbol not occurring in $\CD$. It is easy to see that $\iqc \vdash \CD^S$. Thus, $L \vdash \CD^S$, and by the assumption we have $L \vdash \CD$.

$\bullet$ For the other two cases, with similar reasoning, we have 
\begin{center}
$\ED^S= (B \to A(c)) \to \exists x (B \to A(x))$, \\ 
$\SW^S= (A(c) \to B) \to \exists x (A(x) \to B)$,
\end{center}
where $c$ is a fresh constant not appearing in $A(x)$ and $B$. As $\ED^S$ and $\SW^S$ are valid in $\iqc$, thus $L \vdash \ED$ and $L \vdash \SW$.


\end{proof}
\begin{corollary}\label{cor: shifts}
    The deSkolemization of quantifier shift axioms is the same for structural and Andrews Skolemization.
\end{corollary}

\subsection{Sequent calculi $\LKpp$ and $\LJpp$}
Currently, we lack a well-structured proof system, such as a cut-free or analytic sequent calculus, for $\QFS$, and Theorem \ref{thm: incompleteness} establishes its frame incompleteness (with respect to Kripke frames). This absence of well-behaved semantics highlights the importance of employing proof-theoretic methods to study Skolemization.
A natural question arises: can we prove that $\QFS$ admits Skolemization through alternative approaches? Aguilera and Baaz in \cite{ABaaz} proposed an innovative sequent calculus for first-order (classical or intuitionistic) logic by relaxing the usual quantifier-inference restrictions. 
This calculus, while allowing unsound inferences (thereby lacking local correctness), maintains global correctness. 
We will adopt their sequent calculus as a tool for our investigation. 
To proceed, we first require some definitions. 


\begin{definition}
If a formula in the premise is changed by an inference, we call it an \textit{auxiliary} formula of the inference and the resulting formula in the conclusion is called the \textit{principal} formula.
An inference is a \emph{quantifier inference} if the principal
formula has a quantifier as its outermost logical symbol. 
If an inference yields a strongly quantified formula $Q x A(x)$ from $A(a)$, where $a$ is a free variable, we say that $a$ is the \emph{characteristic variable} of the inference.
For a proof $\pi$, we say $b$ is a \emph{side variable} of $a$ in $\pi$ (written $a<_\pi b$ ) if $\pi$ contains a strong-quantifier inference of one of the forms:
\begin{center}
\begin{tabular}{c c}
\AxiomC{$\Gamma \Rightarrow A(a, b, \vec{c})$}
\UnaryInfC{$\Gamma \Rightarrow  \forall x A(x, b, \vec{c})$}
\DisplayProof
&
\AxiomC{$A(a, b, \vec{c}), \Gamma \Rightarrow \Delta$}
\UnaryInfC{$\exists x A(x, b, \vec{c}), \Gamma \Rightarrow \Delta$}
\DisplayProof
\end{tabular}
\end{center}
\end{definition}

\begin{definition}\cite{ABaaz} \label{Def: suitable} 
\label{Def: suitable quantifier}
A quantifier inference is \textit{suitable} for a proof $\pi$ if either it is a weak-quantifier inference, or the following are satisfied:
\item[1)] (substitutability) the characteristic variable does not appear in the conclusion of $\pi$.
\item[2)] (side-variable condition) the relation $<_\pi$ is acyclic.
\item[3)]
(very weak regularity) if two strong-quantifier inferences in $\pi$ have the same characteristic variable, they must have the same principal formula. 
\end{definition}
For instance, each quantifier inference of $\mathbf{LJ}$ is suitable for every $\mathbf{LJ}$-proof. To see the justification of the choice of these conditions, please see Appendix.

\begin{definition}\cite{ABaaz}
The calculus $\mathbf{LK^{++}}$ (resp. $\mathbf{LJ^{++}}$) is obtained from $\mathbf{LK}$ (resp. $\mathbf{LJ}$), by dropping the restriction preventing the characteristic variable of an inference from
appearing in its conclusion (i.e., the usual eigenvariable condition), and adding the restriction that a proof may only
contain quantifier inferences that are suitable for it.
\end{definition}



It is worth mentioning some facts about these two calculi.
\begin{theorem} \cite{ABaaz} \label{unsound}
The following hold for $\mathbf{LK}^{++}$ and $\mathbf{LJ^{++}}$:
\begin{enumerate}
\item (Correctness) If a sequent is $\LKpp$-derivable, then it is already $\LK$-derivable.
\item There is no elementary function bounding the length of the shortest cut-free $\mathbf{LK}$-proof of a formula in terms of its shortest cut-free $\mathbf{LK}^{++}$-proof.
\item $\mathbf{LJ^{++}}$ is \textit{not} sound for the intuitionistic logic $\mathsf{IQC}$ and does not admit cut elimination.

\end{enumerate}
   \end{theorem}

\begin{example}\label{unsound-ex}
Consider the following $\LJpp$-proof: 
  \begin{align} \label{example2}
\infer[]{ \Rightarrow \exists x\, \big( A(x) \rightarrow A(f(x))\big)}{
\infer[]{ \Rightarrow A(a) \to A(f(a))}{
\infer[]{A(a)  \Rightarrow A(f(a))}{
	\infer[]{A(a)  \Rightarrow \forall x\, A(x)}{A(a) \Rightarrow A(a)}
	&
	\infer[]{\forall x\, A(x)  \Rightarrow A(f(a))}{A(f(a))  \Rightarrow A(f(a))}
}
}
}
\end{align}
If there were a cut-free proof of
\begin{equation} \label{eqcutelimination}
 \Rightarrow \exists x\, \big( A(x) \rightarrow A(f(x))\big),
\end{equation} 
by the subformula property, this derivation would consist entirely of subformulas of \eqref{eqcutelimination} and would contain no strong quantifier inferences. Therefore, it would already be an $\LJ$-proof. This is impossible, as \eqref{eqcutelimination} is not intuitionistically valid. 
\end{example}


As we have seen, $\LJpp$ is not sound for the intuitionistic logic $\mathsf{IQC}$, can we correct the calculi? 
\begin{definition}

We define an underlying sequent calculus $\mathbf{QFS}$ for the logic $\QFS$ 
such that a proof of a sequent $\Pi \Rightarrow \Gamma$ in $\mathbf{QFS}$ corresponds to a proof of a sequent $\Delta, \Pi \Rightarrow \Gamma$ in $\LJ$ where $\Delta$ consists of instances of $\CD$, $\ED$, and $\SW$. 
\end{definition}

Now, we obtain a crucial result: all the quantifier shift principles are provable in $\mathbf{LJ}^{++}$
. In fact:

\begin{restatable}{theorem}{ThmQFSLJpp}
\label{thm: QFS, LJpp}
Let $S$ be a sequent. We have 
\[
\mathbf{LJ^{++}} \vdash S \qquad \text{if{f}} \qquad \mathbf{QFS} \vdash S.
\]
\end{restatable}

\begin{proof}
For a more detailed proof, see Appendix.

$(\Leftarrow)$ Let $\mathbf{QFS} \vdash S$. If $\mathbf{LJ} \vdash S$ then $\mathbf{LJ}^{++} \vdash S$. Moreover, the principles $\CD$, $\ED$, and $\SW$ are provable in $\mathbf{LJ}^{++}$, e.g.
\begin{center}
\begin{tabular}{c c}
\hspace{-20pt}
$\mathbf{LJ}^{++} \vdash \SW$
&
\hspace{20pt}
\AxiomC{$A(a) \Rightarrow A(a)$}
\RightLabel{$*$}
\UnaryInfC{$A(a) \Rightarrow \forall x A(x)$}
\AxiomC{$B \Rightarrow B$}
\BinaryInfC{$A(a), \forall x A(x) \to B \Rightarrow B$}
\UnaryInfC{$\forall x A(x) \to B \Rightarrow A(a) \to B$}
\UnaryInfC{$\forall x A(x) \to B \Rightarrow \exists x (A(x) \to B)$}
\DisplayProof
\end{tabular}
\end{center}
Note the rule marked by $(*)$ makes the proof unsound for $\mathbf{LJ}$.

$(\Rightarrow)$ Let $\mathbf{LJ}^{++} \vdash^{\pi} S$, where $S=(\Sigma \Rightarrow \Lambda)$. Change $\pi$ to $\pi'$ by replacing each unsound universal and existential quantifier inference as follows:
\small \begin{center}
\begin{tabular}{c}
\AxiomC{$\Gamma \Rightarrow A(a)$}
\UnaryInfC{$\Gamma \Rightarrow \forall x A(x)$}
\DisplayProof
$\rightsquigarrow$
\AxiomC{$\Gamma \Rightarrow A(a)$}
\AxiomC{$\forall x A(x) \Rightarrow \forall x A(x)$}
\BinaryInfC{$\Gamma, A(a) \to \forall x A(x) \Rightarrow \forall x A(x)$}
\DisplayProof
\end{tabular}
\end{center}

\small \begin{center}
\begin{tabular}{c}
\AxiomC{$\Gamma, B(b) \Rightarrow \Delta$}
\UnaryInfC{$\Gamma, \exists x B(x) \Rightarrow \Delta$}
\DisplayProof
$\rightsquigarrow$
\AxiomC{$\exists x B(x) \Rightarrow \exists x B(x)$}
\AxiomC{$B(b), \Gamma \Rightarrow \Delta$}
\BinaryInfC{$\Gamma, \exists x B(x), \exists x B(x) \to B(b) \Rightarrow \Delta$}
\DisplayProof
\end{tabular}
\end{center}
\normalsize Denote $A_i(a_i) \rightarrow \forall x A_i(x)$ by $\alpha_i(a_i)$ and $\exists x B_j(x) \rightarrow B_j(b_j)$ by $\beta_j(b_j)$.
Therefore, $\pi'$ is a proof in $\mathbf{LJ}$ such that
\[
\mathbf{LJ} \vdash^{\pi'} \{\alpha_i\}_{i \in I}, \{\beta_j\}_{j \in J}, \Sigma \Rightarrow \Lambda, \tag{1}
\]
for (possibly empty) sets $I$ and $J$. Now, we claim that we can use the rule $(\exists L)$ as many times as needed to get 
\[
\mathbf{LJ} \vdash \{\exists y_i \alpha_i (y_i)\}_{i \in I}, \{\exists z_j \beta_j(z_j))\}_{j \in J}, \Sigma \Rightarrow \Lambda. \tag{2}
\]
Suppose the claim holds. To get a proof of $S$ in $\mathbf{QFS}$, we remove the additional formulas in $(2)$. 
For each $i \in I$ and $j \in J$, we have 
\small\[
\mathbf{QFS} \vdash \ \Rightarrow (\forall y_i A_i(y_i) \to \forall x A_i (x)) \to \exists y_i (A_i(y_i)\to \forall x A_i (x)),
\]
\normalsize since the succedent of the sequent is an instance of the principle $\SW$ and hence provable in $\mathbf{QFS}$. Similarly,
\small \[
\mathbf{QFS} \vdash \ \Rightarrow (\exists x B_j(x) \to \exists z_j B_j(z_j)) \to \exists z_j(\exists x B_j(x) \rightarrow B_j(z_j)),
\]
\normalsize as the succedent is an instance of the principle $\ED$. Therefore, 
\[
\mathbf{QFS} \vdash \ \Rightarrow \exists y_i (A_i(y_i)\to \forall x A_i (x)),
\]
\[\mathbf{QFS} \vdash \ \Rightarrow \exists z_j(\exists x B_j(x) \rightarrow B_j(z_j)).
\]
By the cut rule on the closure 
of the above sequents (under the universal quantifier), and sequent $(2)$, we get $\mathbf{QFS} \vdash S$.
\end{proof}


\begin{definition}
A logic $L$ has the \emph{existence property}, EP, if whenever $L \vdash \exists x A(x)$, where $A(x)$ has no other free variables, then there is some term $t$ such that $L \vdash A(t)$.
\end{definition}

\begin{remark}
    Example \ref{unsound-ex} shows  $\QFS$ does not have EP.  
   \end{remark}
\subsection{From quantifier shifts to deSkolemization}

The following theorem is proved in \cite[Proposition 4.8]{ABaaz} for $\mathbf{LK}^{++}$. We manipulate the proof to get the result for $\LJpp$.

\begin{theorem}\label{thm: cut free LJpp}
Let $\Ga \Rightarrow \De$ be a given sequent. If the Skolemization of the end-sequent $\Gamma \Rightarrow \Delta$ is cut-free $\mathbf{LJ}^{++}$-derivable with atomic axioms, then $\Gamma \Rightarrow \Delta$ is cut-free $\mathbf{LJ}^{++}$-derivable. 
\end{theorem}

\begin{proof}
Let $\pi$ be a cut-free proof with atomic axioms\footnote{Atomic axioms are $\bot\Rightarrow $ and $A\Rightarrow A$, where $A$ is an \emph{atomic} formula.} only, such that 
\[
\LJpp \vdash^\pi (\Gamma \Rightarrow \Delta)^S.
\]
We want to prove that there exists a cut-free proof of $\Gamma \Rightarrow \Delta$ in $\LJpp$. Therefore, we need to deSkolemize $\pi$ and obtain a cut-free proof of the original end-sequent $\Gamma \Rightarrow \Delta$.  \\
\textbf{The deSkolemization algorithm}\footnote{Note that it is a generalization of the second $\varepsilon$-theorem \cite{hilbert}.}:\\
Before we give the algorithm, we should note that for the deSkolemization procedure, it is crucial to have a target end-sequent $\Gamma \Rightarrow \Delta$ (which is chosen under the assumption that the proof of the Skolemization of this target end-sequent is given). 
Note that, by definition, in $(\Gamma \Rightarrow \Delta)^S$ the sets of Skolem functions associated with each formula in $\Gamma \Rightarrow \Delta$ are distinct.
\begin{enumerate}
   \item Replace every Skolem term $t$ 
    with a fresh free variable (indexed with the Skolem term) $a_t$. 
\item We argue on the branches of the proof $\pi$ top down and introduce the strong and weak quantifiers topmost as possible, fitting to the formulas in the target end-sequent. This means that no quantifiers are inferred that
violate the structure of the formulas in the target end-sequent. The topmost
inferences on different branches
may be in any order, but above contractions, all quantifiers that are
possible to infer have to be inferred above contractions. The implicit contractions ($\vee$-left, $\wedge$-right, $\limp$-left) are also
considered as contractions.
\item Propositional inferences beside weakenings, $\vee$-left, $\wedge$-right, and contractions are performed as before. In case of weakenings, $\vee$-left, $\wedge$-right and contractions, 
the subformulas containing quantifiers not inferred are
replaced by the inferences
of the corresponding deSkolemized subformulas according to the target end-sequent.
\item Repeat the procedure until the formulas up to quantifiers are completed. 
\end{enumerate}
The procedure terminates. Note that there is no obstacle to
contractions according to steps 2) and 3). In this way, the propositional proof structure remains the same. 

Now, we need to show that the obtained deSkolemized proof is an $\LJpp$-proof.
Therefore, we need to check that the suitable conditions in the Definition \ref{Def: suitable quantifier} are satisfied:

\item[$\bullet$] \textbf{Substitutability} is fulfilled as after the deSkolemization procedure, the end-sequent will contain no Skolem terms, therefore a new characteristic variable does not occur in the conclusion.
\item[$\bullet$] \textbf{Very weak regularity} is obviously satisfied, as always the same Skolem terms generate the same variables and therefore belong to the same formula. This may not be possible in $\LJ$ (since while introducing strong quantifiers it may fail for eigenvariable condition).
\item[$\bullet$] \textbf{Side variable condition} - the order of the quantifier inferences are determined by $\Gamma\Rightarrow \Delta$. We set it such that the characteristic variable of the later quantifier inference is always bigger than the characteristic variable of any earlier quantifier inferences. The side variable conditions are part of this order and therefore acyclic.
  \end{proof}

\begin{remark}
Note a few key remarks in the above proof:
\item[$\bullet$] We cannot eliminate cuts in an $\LJpp$-proof. To avoid cut elimination, we take an $\LJ$-proof with cuts of the end-sequent $\Gamma \Rightarrow \Delta$. We Skolemize everything (including the end-sequent) except cuts and eliminate cuts. A cut-free $\LJ$-proof with only weak quantifiers in the end-sequent obtained in this way is also cut-free $\LJpp$-proof $\pi$ of $(\Gamma \Rightarrow \Delta)^S$ (see the proof of Theorem \ref{Thm: Skolem for QFS}). 
\item[$\bullet$] As we have already seen in Theorem \ref{unsound}, $\LJpp$ is not sound for $\iqc$ and so it can derive unsound results as in Example \ref{unsound-ex}. Therefore, it is important to guarantee that the calculi can be corrected, and this is shown in the Theorem \ref{thm: QFS, LJpp}. 
 \item[$\bullet$] Restriction to atomic axioms is purely for simplicity reason since deSkolemizing atomic axioms produce axioms again, which is not the case for non-atomic axioms. In $\LJ$ and $\LJpp$ we can always derive compound 
axioms from atomic ones. 
\end{remark}
\noindent We illustrate the deSkolemization algorithm with an example.
\begin{example}
 Consider an instance of $\CD$
\[
\forall x(A(x)\lor B)\Rightarrow (\forall xA(x)\lor B),
\]
where $A$ and $B$ are atomic formulas, $A$ has only the variable $x$, and $B$ has no variables. The following is an $\LJ$-proof of the Skolemization of $\CD$, where $c$ is a fresh variable:
\begin{center}
\AxiomC{$A(c) \Rightarrow A(c)$}
\UnaryInfC{$A(c) \Rightarrow A(c)\lor B$}
\AxiomC{$B\Rightarrow B$}
\UnaryInfC{$B \Rightarrow A(c)\lor B$}
\BinaryInfC{$A(c)\lor B \Rightarrow A(c)\lor B$}
\UnaryInfC{$\forall x (A(x)\lor B) \Rightarrow A(c)\lor B$}
\DisplayProof
\end{center}
This is also a cut-free proof in $\LJpp$.
Now, we deSkolemize the above proof with respect to the end-sequent $\CD$:
\begin{center}
\AxiomC{$A(a_c) \Rightarrow A(a_c)$}
\RightLabel{\scriptsize $(*)$}
\UnaryInfC{$A(a_c) \Rightarrow \forall xA(x)$ }
\UnaryInfC{$A(a_c) \Rightarrow \forall xA(x) \vee B$ }
\AxiomC{$B\Rightarrow B$}
\UnaryInfC{$B \Rightarrow \forall x A(x) \vee B$}

\BinaryInfC{$A(a_c)\lor B \Rightarrow \forall x A(x)\lor B$ }
\RightLabel{\scriptsize $(\dagger)$}
\UnaryInfC{$\forall x(A(x)\lor B) \Rightarrow \forall x A(x)\lor B$ }
\DisplayProof
\end{center}
where $a_c$ is a fresh variable for the Skolem term $c$. Call this proof $\pi$. In the left branch, the topmost sequent where $a_c$ is appeared is designated by $(*)$ and in the right branch by $(\dagger)$. Clearly, $\pi$ is a cut-free $\LJpp$-proof of $\CD$.

Now, as an example of how Theorem \ref{thm: QFS, LJpp} works, we will transform this $\LJpp$-proof to a proof in $\mathbf{QFS}$. As there is only one unsound quantifier inference ($*$), we replace it by
\begin{center}
\AxiomC{$A(a_c) \Rightarrow A(a_c)$}
\AxiomC{$\forall x A(x) \Rightarrow \forall x A(x)$}
\BinaryInfC{$A(a_c), A(a_c) \to \forall x A(x) \Rightarrow \forall x A(x)$}
\DisplayProof
\end{center}
We get:
\[
\LJ \vdash A(a_c) \vee B, A(a_c) \to \forall x A(x) \Rightarrow \forall x A(x) \vee B.
\]
Call this proof $\sigma$. Thus,
\begin{center}
\AxiomC{$\sigma$}
\noLine
\UnaryInfC{$A(a_c) \vee B, A(a_c) \to \forall x A(x) \Rightarrow \forall x A(x) \vee B$}
\RightLabel{\scriptsize $(\forall L)$}
\UnaryInfC{$\forall y (A(y) \vee B), A(a_c) \to \forall x A(x) \Rightarrow \forall x A(x) \vee B$}
\RightLabel{\scriptsize $(\exists L)$}
\UnaryInfC{$\forall y (A(y) \vee B), \exists z(A(z) \to \forall x A(x)) \Rightarrow \forall x A(x) \vee B$}
\DisplayProof
\end{center}
Note that all the eigenvariable conditions are satisfied, therefore, this is a proof in $\LJ$ and hence in $\mathbf{QFS}$. However, as shown in the proof of Theorem \ref{thm: QFS, LJpp}, we have
\[
\mathbf{QFS} \vdash \; \Rightarrow \exists z(A(z) \to \forall x A(x)).
\]
Thus, by cut, we get $\mathbf{QFS} \vdash \CD$.
\end{example}

Another illuminating example is given in Example \ref{Ex: deSkolem}, which because of lack of space we have moved it to Appendix.

Now, we aim to prove that if an intermediate logic contains $\QFS$, then Skolemization is sound and complete for it. The following proposition is an easy observation.

\begin{proposition}\label{Prop: QFL}
Let $L$ be an intermediate logic and $\Gamma \Rightarrow \Delta$ be a sequent where $\Gamma$ and $\Delta$ are multisets of formulas in $L$. If $\bigwedge \Gamma \to \bigvee \Delta \in L$, then $\LJ \vdash \Pi, \Gamma \Rightarrow\Delta$, where $\Pi$ is a multiset consisting of formulas from $L$. 
\end{proposition}

\begin{theorem}\label{Thm: Skolem for QFS}
Let $L$ be an intermediate logic. If $\QFS \subseteq L$, then Skolemization is sound and complete for $L$.
\end{theorem}

\begin{proof} 
Let $L$ be an intermediate logic such that $\QFS \subseteq L$. We have to show that for any formula $A$, we have
\[
L \vdash A \Longleftrightarrow L \vdash A^S.
\]
\item[$(\Rightarrow)$] Let $L$ be an intermediate logic. Using the following facts
\begin{center}
$L \vdash \forall x B(x, \bar{y}) \to B(t , \bar{y})$ \quad 
$L \vdash B(t , \bar{y}) \to \exists x B(x, \bar{y})$,
\end{center}
for any formula $B$ and term $t$, we can easily see that Skolemization is sound for $L$. Thus, for any formula $A$:
\[
\text{if} \quad L \vdash A \quad \text{then} \quad L \vdash A^S. \tag{$*$}
\] 

\item[$(\Leftarrow)$] We aim to show that if $\QFS \subseteq L$ and $L \vdash A^S$, then $L \vdash A.$
Assume $L \vdash A^S$. Therefore, by Proposition \ref{Prop: QFL}, there are formulas $B_1, \dots, B_n \in L$ such that 
\[
\LJ \vdash B_1, \dots, B_n \Rightarrow A^S.
\]
 As shown in \cite{BaazIemhoff06}, if $\LJ \vdash S$ for a sequent $S$, then $\LJ \vdash S^S$. Thus,
 \[
\LJ \vdash (B_1, \dots, B_n \Rightarrow A^S)^S,
\]
 which is by definition 
 \[
\LJ \vdash (B_1, \dots, B_n \Rightarrow A)^S.
\]
Therefore, by the cut-elimination of $\LJ$, there is a cut-free proof $\pi$ such that
\[
\mathbf{LJ} \vdash^\pi (B_1 , \dots , B_n \Rightarrow A)^S.
\]
As $\pi$ is also an $\LJpp$-proof and cut-free, by Theorem \ref{thm: cut free LJpp} 
\[
\LJpp \vdash B_1 , \dots , B_n \Rightarrow A.
\]
By theorem \ref{thm: QFS, LJpp}, we get
\[
\mathbf{QFS} \vdash B_1 , \dots , B_n \Rightarrow A.
\]
As each $B_i \in L$ and $\QFS \subseteq L$ we conclude that $L \vdash A$.
\end{proof}

\begin{corollary}
   \label{cor: Skolem for QFS}
Let $L$ be an intermediate logic. If $\QFS \subseteq L$, then Andrews Skolemization is sound and complete for $L$. 
\end{corollary}
\begin{proof}
  By corollary \ref{cor: shifts} and the fact that the deSkolemization of Andrews Skolemization leads to the same argument as in the proof of Theorem \ref{Thm: Skolem for QFS}.  Note that Andrews Skolemization also preserves the very weak regularity.
\end{proof}

This concludes the proof of our main Theorem \ref{mainth} that states that for any intermediate logic, admitting structural or Andrews Skolemization is equivalent to containing all quantifier shift principles.
\begin{remark}
    It is worth mentioning that finite-valued first-order G\"{o}del logics admit deSkolemization, which follows from the fact that these logics allow all classical quantifier shift principles (\cite{MatthiasZach}). 
Moreover, the logic $\QFS + \mathrm{Lin}$, and the logic determined by the class of all finite linearly ordered Kripke frames with constant domain studied in \cite{OnoProblems} are concrete examples of intermediate logics to which our Theorem \ref{Thm: Skolem for QFS} applies. For a map of logics containing $\QFS$, see \cite{umezawa1959}, where $\QFS$ is denoted by $\mathrm{LFGP}_2$, particularly Theorem 18.
\end{remark}

\section{Parallel Skolemization } \label{sec: SL}
We established Skolemization for intermediate logics with quantifier shifts. In this section, we show for certain logics outside this category, where standard and Andrew's Skolemization fail (as in most intermediate logics), an alternative form of Skolemization may still be admissible \cite{BaazIemhoff16,BaazLolic_interpolation,CINTULA201954}. See \cite{Rose} for a general view on alternative Skolemization methods. 


\begin{definition}
For a closed first-order formula $A$, \emph{parallel Skolemization} $A^\SP$ is defined similarly to the structural Skolemization \ref{Def: St Skolem}.
The only difference is the following: 
for all strong quantifier occurrences
$A_{-(Q y)}$ is replaced by $\bigvee_{i = 1}^{n} C(f_i (\overline{x}))$ for negative occurrence of $\exists$ of the form $\exists x C(x)$  and by $\bigwedge_{i = 1}^{n}C(f_i (\overline{x}))$ for positive occurrence of $\forall$ of the form $\forall x C(x)$
after omission of $(Q y)$, where 
$f_i$ are new function symbols, $\overline{x}$ are the weakly quantified variables of the scope, and $n$
is the degree \footnote{$n$ iterates according to the branch corresponding to the Kripke model (see, \cite{BaazIemhoff16},\cite{CINTULA201954}). Note also that every branch of a Kripke model of the logic has its own Skolem function.} of the parallel Skolemization. 
\end{definition}
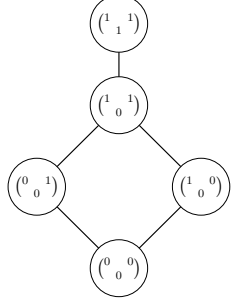
\begin{figure}
\centering
\newcommand\mypic{
\begin{tikzpicture}
\Vertex[L=$ 1 \quad 1 \choose 1$,x=1,y=0] {Y}; 
\Vertex[L=$ 1 \quad 1 \choose 0$,x=1,y=-2]{X}; 
\Vertex[L=$ 0 \quad 1 \choose 0$,x=-1,y=-4]{A};
\Vertex[L=$ 1 \quad 0 \choose 0$,x=3,y=-4]{B};
\Vertex[L=$ 0 \quad 0 \choose 0$,x=1,y=-6]{D} ;

\Edges [](Y,X) 
\Edges[](A,X,B)
\Edges[](A,D,B)
\end{tikzpicture}}
\caption{ \scriptsize{The lattice-based 5-valued logic.}} 
\label{lattice}
\resizebox{3cm}{!}{\mypic} 

\end{figure}

\begin{example}\label{allsk}

Consider the formula
$$\varphi\colon \neg( \forall x \forall y (\exists z P(x,z) \lor Q(y,x))) .$$
The structural Skolemization of $\varphi$ is
$$\neg(\forall x \forall y (P(x,f(x,y)) \lor Q(y,x)) ).$$
 Its Andrews Skolemization is
$$\neg(\forall x \forall y (P(x,g(x)) \lor Q(y,x)) ).$$
Its parallel Skolemization, where $n=2$  is 
\begin{align*}
    \neg(\forall x \forall y (P(x,f_1(x,y)) \lor P(x,f_2(x,y))\lor Q(y,x))).
\end{align*}
\end{example}

\begin{definition}\label{1-entail} 
Let $V$ be a truth value set with one designated truth value $1$, $\Gamma$ a (possibly infinite) set of formulas, and $A$ be a formula. We say $\Gamma$ \emph{entails $A$ w.r.t. $V$}, denoted by $\Gamma \VDash_V A$ whenever
\[
\forall I\ \  \forall B \in \Gamma \quad \text{if} \quad I(B) = 1 \quad \text{then} \quad I(A) = 1,
\]
where $I$ is the interpretation on $V$. If $\Gamma \VDash_V A$ for any valuation $V$, then we say that $\Gamma$ \emph{entails} $A$, denoted by $\Gamma \VDash A$.
\end{definition}
The following example provides a logic
falsifying $\SW$, therefore, no structural or Andrews Skolemization is admissible. However, it admits parallel Skolemization.
\begin{example}\label{SL}
Let us consider a 5-valued lattice corresponding to constant-domain 3-world Kripke frames.
We look at the truth valuations of propositional variables in constant-domain 3-world Kripke frames. Denote the domain by $D$ and $D\neq \emptyset$.
The valuations are 
 $\alpha_{(a,b,c)} \equiv {a \quad b \choose c}$, for any $a,b,c \in \{0,1\}$
with the designated value   $\alpha_{(1,1,1)}\equiv 1$.
We order the truth values according to the following: 

\begin{center}
{$\displaystyle\alpha_{(a,b,c)} \leq \alpha_{(a',b',c')} \quad \text{ iff } \quad a\leq a', b\leq b' \text{ and }  c\leq c'$.}
\end{center}

We obtain the lattice-based 5-valued logic depicted in Fig.\ref{lattice}. We denote it by $\CD_5$.
This logic verifies, under any interpretation, the corresponding true formulas in the Kripke frames and the other way around. The valuations are as follows: 
\begin{align*}   
I(\bot)= \scriptstyle{{ 0 \quad 0 \choose 0}}\equiv 0\\
I(A\lor B)=\text{sup}\{I(A),I(B)\}\\     I(A \wedge B)=\text{inf}\{I(A),I(B)\}\\     I(\exists xA(x))=\text{sup}\{I(A(d)) \mid d\in D\}\\     I(\forall xA(x))=\text{inf}\{I(A(d))\mid d\in D\}\\
I(A\rightarrow B)= 1 \text{ iff }\{I(A)\leq I(B)\}
\end{align*}

It is obvious that in the logic $\CD_5$, the quantifier shift principles $\SW$ and $\ED$ do not hold; see the model in Example \ref{Sle} for $\SW$, where
\begin{align*} 
\nVDash \exists x(P(x)\limp \forall y P(y)).
\end{align*}

As for parallel Skolmization, we replace negative (resp. positive) occurrences of $\exists$ (resp. $\forall$) with disjunction (resp. conjunction) of two Skolem functions 
\begin{align*}
    \exists xA(x,\bar{y})\Rightarrow A(f(\bar{y}),\bar{y})\lor A(g(\bar{y}),\bar{y})\\
    \forall xA(x,\bar{y})\Rightarrow A(f(\bar{y}),\bar{y})\land A(g(\bar{y}),\bar{y}).
\end{align*}
It is easy to see that the parallel Skolemization for branching degree $n=2$ is valid in $\CD_5$ since it simulates $\sup$ and $\inf$ of the lattice; let  
\begin{align*}
   I(\exists x A(\overline{c},x))=\alpha_{(1,1,0)}
\end{align*}
but for all elements $d\in D$ 
\begin{align*}
I( A(\overline{c},d))\neq\alpha_{(1,1,0)}.
\end{align*}
Therefore, there are elements $e, f$ such that 
 \begin{align*}
     I( A(\overline{c},e))=\alpha_{(1,0,0)} \text{ and } I( A(\overline{c},f))=\alpha_{(0,1,0)}.
 \end{align*}
 We assign $e$ to $g_1(\overline{c})$ and $f$ to $g_2(\overline{c})$.
Then, 
\begin{align*}
I(A(\overline{c},g_1(\overline{c})\vee A(\overline{c},g_2(\overline{c}))= \text{sup}\{I(A(\overline{c},g_1(\overline{c})),I(A(\overline{c},g_2(\overline{c}))\} \\ =\alpha_{(1,1,0)}=I(\exists x A(\overline{c},x)).
\end{align*}
We have shown above the case for  ($\sup, {1 \quad1 \choose 1},\exists, \vee$). Similar argument applies to ($\inf, {0 \quad0 \choose 0},\forall, \wedge$). 
Therefore, the Skolemization does not change the outcome of the interpretation.



\end{example}

\begin{proposition}
 We have   $\CD_5\VDash A^\SP$  if{f} $\CD_5\VDash A$.
\end{proposition}
\begin{proof}
The proof is an immediate consequence of Lemma 5.1 in \cite{BaazLolic_interpolation} for lattice-based finitely valued logics corresponding to finite propositional and constant-domain Kripke frames.
\end{proof}




\section{Frame-incompleteness}\label{semantic-proof} 

In the previous sections, we have shown that $\QFS$ determines whether a logic has Skolemization or not. In this section, we investigate the logic $\QFS$ further. We characterize the class of frames on which $\mathsf{QFS}$ is valid, but also show that the logic is frame-incomplete.

\begin{definition}
Let $F=(W, \preccurlyeq, D)$ be a frame. Define:

\noindent ($\WF$)
$\forall w \in W$ 
the set $\preccurlyeq [w]$ is well-founded.

\noindent ($\cWF$)
$\forall w \in W$ 
the set $\preccurlyeq [w]$ is conversely well-founded.

\noindent ($\FDS$)
$\forall u,v,w \in W \ \text{if} \  u \neq v \ u \succcurlyeq w \ v \succcurlyeq w \ \text{then} \ D_u=D_v=D_w , \ |D_w|=1$

The latter condition, called \emph{Fork-Domain Singleton}, states that if there are three worlds in a frame that form a fork, then their domains are the same and singleton. 

\end{definition}
\begin{remark}\label{rem: WF, cWF}
Note that a frame $F$ satisfying $\WF$ does not imply that $\preccurlyeq$ is well-founded on $W$. For instance, if $W$ only consists of a descending chain $v_0 \succcurlyeq v_1 \succcurlyeq \dots$, then $F$ satisfies both $\WF$ and $\cWF$. In particular, for any $w \in W$, the set $\preccurlyeq [w]$ is well-founded and conversely well-founded. However, $\preccurlyeq$ is not well-founded on $F$.
\end{remark}


The following definition introduces rich classes of frames for $\QFS$ and its fragments.

\begin{definition}\label{Dfn: class F}
Define $\calF$ as the class of the following Kripke frames $F$ closed under the disjoint union:

$F$ is constant domain with domain $\calD$, satisfying one of the following conditions:
\begin{enumerate}
\item 
$|\calD|=1$,
\item 
$|\calD|>1$, $\calD$ is finite, and $F$ is linear,
\item 
$\calD$ is infinite, $F$ is linear, and satisfies both $\WF$ and $\cWF$. 
\end{enumerate}
Define $\calF_{\ED}$ (resp. $\calF_{\SW}$) as the class of Kripke frames containing $\calF$ and also frames of the form

$\bullet$ constant domain, linear, infinite domains, satisfying $\WF$ (resp. satisfying $\cWF$).
\end{definition}


\subsection{Frame characterization}
In this subsection, we state and prove the following theorem that offers an elegant characterization of the class of frames for $\QFS$ and its fragments. Note that, as is standard in the literature, when the forcing of a predicate is not explicitly specified at a world in a Kripke model, it is understood to not be forced at that world.
\begin{theorem}[Frame characterization]\label{Thm: soundness}
Let $F$ be a frame. 
\begin{itemize}
\item 
$F \vDash \CD \quad$ if and only if $\quad F$ is constant domain.
\item 
$F \vDash \ED \quad$ if and only if $\quad F \in \calF_{\ED}$.
\item 
$F \vDash \SW \quad$ if and only if $\quad F \in \calF_{\SW}$.
\item 
$F \vDash \QFS \ \ $ if and only if $\quad F \in \calF$.
\end{itemize}
\end{theorem}

The remainder of this subsection is devoted to the proof of the theorem. To enhance accessibility for the reader, we begin by establishing the lemmas required for the $(\Rightarrow)$ direction of all items in the theorem, followed by the $(\Leftarrow)$ direction. 

\noindent \textbf{Towards the proof of $(\Rightarrow)$ of Theorem \ref{Thm: soundness}:} The following lemma shows that being constant domain is a necessary condition for a frame to be a frame for $\CD$, $\ED$, or $\SW$, thereby implying that the same holds for $\QFS$.

\begin{restatable}{lemma}{LemConstantDomain}\label{lem: constant domain}
Let $F$ be a frame.
\begin{enumerate}
    \item 
    If $F \vDash \CD$, then $F$ is constant domain.
    \item 
     If $F  \vDash \ED$, then $F$ is constant domain.
     \item 
      If $F \vDash \SW$, then $F$ is constant domain.
\end{enumerate}
\end{restatable}
\begin{proof}
    See Appendix.
\end{proof}

As a corollary of Lemma \ref{lem: constant domain} we have:
\begin{corollary}\label{Cor: constant domain}
Let $F$ be a frame. If $F \vDash \QFS$, then $F$ is constant domain.
\end{corollary}

The following lemma states that a frame with a fork and a world of domain size two or more cannot be a frame for $\ED$ or $\SW$. This means that in such case, linearity is a necessary condition for a frame to be a frame of $\QFS$. 

\begin{restatable}{lemma}{LemFDS} \label{lem: FDS}
Let $F$ be a frame. 
\begin{enumerate}
\item
If $F \vDash \ED$, then $F$ satisfies $\FDS$.
\item 
If $F \vDash \SW$, then $F$ satisfies $\FDS$.
\end{enumerate}

\end{restatable}

\begin{proof}
See Appendix.
\end{proof}
If we combine Corollary \ref{Cor: constant domain} and Lemma \ref{lem: FDS} we get:

\begin{corollary} \label{Cor: linear or singleton}
Let $F$ be a frame. If $F \vDash \ED$ (or $F \vDash \SW$), then either $F$ is linear or $\forall w \in W$ we have $|D_w|=1$.
\end{corollary}

The following lemma states that if a frame validates $\ED$ (resp. $\SW$) then WF (resp. cWF) is satisfied in it.
\begin{restatable}{lemma}{LemWF} \label{lem: WF}
Let $F$ be a frame such that there exists a world $w \in W$ where $D_w$ is countably infinite.
\begin{enumerate}
\item 
If $F \vDash \ED$, then  $F$ satisfies $\WF$.
\item 
If $F \vDash \SW$, then  $F$ satisfies $\cWF$.
\end{enumerate}
\end{restatable}

\begin{proof} If $F \vDash \ED$ or $F \vDash \SW$, then by Lemma \ref{lem: constant domain} it is constant domain. By the assumption, the domain is countably infinite. Call it $\calD=\{a_0, a_1, \dots\}$. Thus, by Corollary \ref{Cor: linear or singleton}, $F$ is linear. We prove each case by contraposition.  Let $Q$ be a nullary predicate and $P$ a unary one. 
\item[1)] 
Suppose $F$ does not satisfy $\WF$. Thus, there exists $w \in W$ where $\preccurlyeq [w]$ contains an infinite descending chain $v_0 \succcurlyeq v_1 \succcurlyeq v_2 \succcurlyeq \dots$. Call this chain $C$. Construct the model $\calM$:


\begin{align*}
\calM , u & \Vdash Q \quad \text{iff} \ \  u \succcurlyeq v_0 \ \text{or} \ u \in C \\
\forall d \in \calD \ \ \calM , v_0 & \Vdash P(d) \\ 
\forall d \in \calD \ \ \calM , w & \nVdash P(d)   \\ 
\forall d \in \calD - \{a_0, \dots, a_{i-1}\} \ \ \calM, u & \Vdash P(d) \qquad i \geq 1, \forall u \succcurlyeq v_i
\end{align*}


\noindent depicted as:
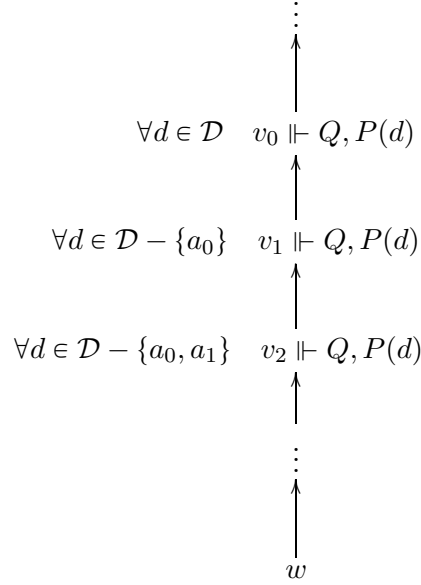
\small\[
  \xymatrix{ 
  \vdots  \\
            \hspace{-15pt}   \forall d \in \mathcal{D} \quad   v_0 \forc Q, P(d) \ar[u]   \\
           \hspace{-45pt}   \forall d \in \mathcal{D} - \{a_0\} \quad   v_1 \forc Q, P(d)  \ar[u]    \\
         \hspace{-58pt}   \forall d \in \mathcal{D} - \{a_0, a_1\} \quad   v_2 \forc Q, P(d)  \ar[u]    \\
           \vdots \ar[u] \\
           w \ar[u]
           }
\]

\normalsize For any world $u \succcurlyeq w$ there are three cases:

$\bullet$
$u \in C$. Thus, there exists $i \geq 0$ such that $u=v_i$. Hence, $\calM, v_i \forc Q, P(a_i)$.

$\bullet$
$u \succcurlyeq v_0$. Thus, $\calM, u \forc Q, P(a_0)$.

$\bullet$
Otherwise. Hence, $\calM, u \nVdash Q$.

\noindent Therefore, we get $\calM, w \Vdash Q \to \exists x P(x)$. However, as $\calM, w \nVdash Q \to P(a_i)$ for any $i \geq 0$, we get $\calM, w \nVdash \exists x (Q \to P(x))$. Thus,
\[
\calM, w \nVdash (Q \to \exists x P(x)) \to \exists x (Q \to P(x)).
\]
\item[2)] See Appendix. 
\end{proof}



Putting Lemmas \ref{lem: constant domain}, \ref{lem: WF} and Corollary \ref{Cor: linear or singleton} together, we get the $(\Rightarrow)$ direction of the theorem.


\begin{example}\label{Ex: two frames}
Take the following frames: 
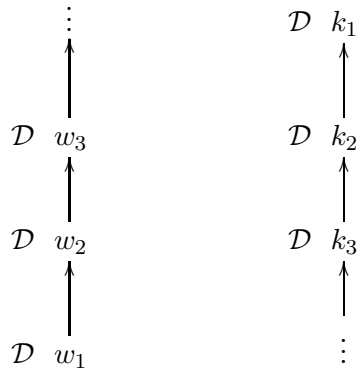
\small \[  \xymatrix{           \vdots &  & &            \hspace{-15pt}  \mathcal{D} \ \ k_1 \\ \hspace{-15pt}  \mathcal{D} \ \ w_3  \ar[u] &  & & \hspace{-15pt}  \mathcal{D} \ \ k_2  \ar[u] \\ \hspace{-15pt}  \mathcal{D} \ \ w_2  \ar[u] & & &   \hspace{-15pt}  \mathcal{D} \ \ k_3  \ar[u] \\ \hspace{-15pt}  \mathcal{D} \ \ w_1 \ar[u] &  & &  \vdots \ar[u]           }\]
\normalsize It is interesting and easy to see that the right one is a frame for $\QFS$. However, the left frame is a frame for $\ED$ and not a frame for $\SW$. Note that if one puts one node below all other nodes in the frame on the right, it no longer is a model of $\QFS$, as it refutes $\ED$.
\end{example}
\noindent \textbf{Towards the proof of $(\Leftarrow)$ of Theorem \ref{Thm: soundness}:}

It is not enough for a frame to be constant domain to validate $\ED$ or $\SW$. However, being constant domain makes a frame sound and complete for the axiom $\CD$, as indicated in Lemma \ref{lem: constant domain} and the following lemma. 
\begin{restatable}{lemma}{LemConsDomCD}\label{lem: cons dom validates CD}
Let $F$ be a constant domain frame. Then $F$ validates $\CD$.
\end{restatable}

\begin{proof} Similar to \cite[Lemma 14.4.]{Mints} and \cite{CD}. See Appendix.
\end{proof}


The following three lemmas provide sufficient conditions on the domains of a frame under which $\QFS$ holds in it. In the first lemma the domains are singletons, with no further restrictions on the frame. In the second the domains are finite, but in addition the frame is required to be constant domain and linear. In the third lemma the frames are constant domain and finite, and it is also shown that if finiteness is weakened to $\WF$ ($\cWF$), then the frame may not be a frame of $\QFS$ but it is a frame of $\ED$ ($\SW$).

\begin{restatable}{lemma}{LemSingleton}\label{lem: singleton domain}
Let $F$ be a frame. If the domain of every world in $F$ is a singleton, then $F \vDash \QFS$.
\end{restatable}
\begin{proof}
Let $F=(W, \preceq, D)$, where $|D_w|=1$ for every $w \in W$. Call this element $a$.

\item[$\bullet$]
As $F$ is constant domain, by Lemma \ref{lem: constant domain}, we get $F \vDash \CD$.

\item[$\bullet$]
To prove $F \vDash \SW$, suppose 
\[
F \vDash \forall x A(x) \to B,
\]
for any formulas $A$ and $B$, where $x$ is not free in $B$. This means that for any model $\calM$ based on $F$ and any world $w$, we have $\calM, w \Vdash \forall x A(x) \to B$, which implies 
\[
\forall w' \succcurlyeq w \quad \text{if} \quad \calM, w' \Vdash \forall x A(x) \quad \text{then} \quad 
 \calM, w' \Vdash B.
\]
However, as the domain is $\{a\}$, we get
\[
\forall w' \succcurlyeq w \quad \text{if} \quad \calM, w' \Vdash  A(a) \quad \text{then} \quad 
 \calM, w' \Vdash B,
\]
which means $\calM, w \Vdash A(a) \to B$. Therefore, $\calM, w \Vdash \exists x A(x) \to B$. Thus, $F \vDash \SW$.
\item[$\bullet$]
$F \vDash \ED$ is similar. See Appendix.
\end{proof}

The following lemma indicates that if the domains of all the worlds are finite in a frame, it is a frame for $\QFS$.

\begin{restatable}{lemma}{LemFiniteDom}\label{lem: finite domain}
    Let $F$ be a linear constant domain frame, where the domain of each world is finite. Then, $F \vDash \QFS$. 
\end{restatable}
\begin{proof}
Let $F$ be a linear constant domain frame, where the domain is $\calD=\{a_1, a_2, \dots, a_n\}$, for some natural number $n$. By Lemma \ref{lem: cons dom validates CD}, $F \vDash \CD$. Now, suppose $F \nvDash \ED$. Hence, there exists a model $\calM$ and a world $w \in W$ such that 
\[
\calM, w \Vdash Q \to \exists x P(x) \qquad 
\calM, w \nVdash \exists x (Q \to P(x)) \tag{1}
\]
for formulas $P$ and $Q$, where $x$ is not free in $Q$. The right statement implies $\calM, w \nVdash Q \to P(a_i)$, for each $1 \leq i \leq n$. Thus,
\[
\exists v_i \geq w \quad \calM, v_i \Vdash Q \quad \calM, v_i \nVdash P(a_i).
\]
As $F$ is linear, the minimum of the set $\{v_1, v_2, \dots, v_n\}$ exists. W.l.o.g. suppose the minimum is $v_1$. Hence, 
\[
\calM, v_1 \Vdash Q \qquad \forall 1 \leq i \leq n \ \calM, v_1 \nvDash P(a_i),
\]
which implies
\[
\calM, v_1 \nVdash Q \to \exists x P(x),
\]
which is a contradiction with (1). Therefore, $F \vDash \ED$.

The case for $\SW$ is similar with the difference that this time we reason by taking the maximum of the worlds. See Appendix.
\end{proof}

The following lemma investigates the case where we relax the condition of the domain being finite.

\begin{restatable}{lemma}{lemWFvalidatesED} \label{lem: WF validates ED}
Let $F$ be a linear constant domain frame.
\begin{enumerate}
\item 
If $F$ satisfies $\WF$, then $F \vDash \ED$.
\item 
If $F$ satisfies $\cWF$, then $F \vDash \SW$.
\item 
If $F$ satisfies both $\WF$ and $\cWF$, then $F \vDash \QFS$.
\end{enumerate}
\end{restatable}
\begin{proof}
Let $F$ be a linear constant domain frame. Denote the domain by $\calD$.

\item[1)] Let $F$ satisfy $\WF$. We want to prove $F \vDash \ED$, i.e.,
\[
F \vDash (B \to \exists x A(x)) \to \exists x (B \to A(x)),
\]
for any formulas $A(x)$ and $B$, where $x$ is not free in $B$. Thus, for any model $\calM$ based on $F$ and any $w \in W$, we will show
\[
\text{if} \quad \calM, w \Vdash B \to \exists x A(x) \quad \text{then} \quad \calM, w \Vdash \exists x (B \to A(x)).
\]
Suppose $\calM, w \Vdash B \to \exists x A(x)$, i.e., for any $v \succcurlyeq w$
\[
\text{if} \quad \calM, v \Vdash B \quad \text{then} \quad \calM, v \Vdash \exists x A(x).
\]
There are two cases:

$\bullet$ $\calM, v \nVdash B$ for all $v \succcurlyeq w$. Take $d \in \mathcal{D}$ arbitrarily. Hence, $\calM, w \Vdash B \to A(d)$, which means $\calM, w \Vdash \exists x (B \to A(x))$.

$\bullet$ $\calM, v \Vdash B$ for some $v \succcurlyeq w$. As $F$ satisfies $\WF$ the set $\preccurlyeq [w]$ is well-founded. Take $S=\{v \succcurlyeq w \mid \calM, v \Vdash B\}$. Clearly $S \subseteq \preccurlyeq [w]$, hence, it has a minimal element $u$. Thus, 
\[
\calM, u \Vdash B \quad \text{hence} \quad \calM, u \Vdash \exists x A(x).
\]
Thus, there exists $d \in \mathcal{D}$ such that $\calM, u \Vdash A(d)$. We will show $\calM, w \Vdash B \to A(d)$. The reason is that for any world $w'$ such that $w \preccurlyeq w' \preccurlyeq u$ we have
\[
\calM, w' \nVdash B \quad \text{hence} \quad \calM, w' \Vdash B \to A(d).
\]
Moreover, for $w' \succcurlyeq u$, as $\calM, u \Vdash B \to A(d)$ we have
\[
\calM, w' \Vdash B \to A(d).
\]
Hence, by definition $\calM, w \Vdash B \to A(d)$. As $F$ is constant domain, $d \in D_w$. Thus, $\calM, w \Vdash \exists x (B \to A(x))$.


\item[2)] Similar to the previous case. See Appendix. 

\item[3)] Clear by the previous parts. In particular, if $F$ is finite, then it satisfies both $\WF$ and $\cWF$. 
\end{proof}

Putting Lemmas \ref{lem: cons dom validates CD}, \ref{lem: singleton domain}, \ref{lem: finite domain}, \ref{lem: WF validates ED} together, we get the $(\Rightarrow)$ direction of the theorem.

\subsection{Kripke frame-incompleteness}
In this subsection, we establish an intriguing result showing that the logic $\QFS$ and some fragments are Kripke frame-incomplete. The theorem highlights a limitation of semantics in deriving deeper results about the logic, suggesting the need to turn to proof-theoretic methods to address questions such as the Skolemization problem.

\begin{definition}\label{Dfn: sound-complete}
The logic $L$ is \emph{sound and complete} w.r.t. the class $\mathcal{C}$ of Kripke frames when for any formula $\varphi$
\[
\phi \in L \quad \text{if and only if} \quad \mathcal{C} \vDash \varphi.
\]
\end{definition}
\begin{definition}\label{Dfn: frame-incomplete}
The logic $L$ is called \emph{frame-incomplete} if there is no class $\mathcal{C}$ of Kripke frames such that $L$ is sound and complete with respect to $\mathcal{C}$. 
\end{definition}

Now, we prove Theorem \ref{thm: incompleteness}, which proves the incompleteness of $\QFS$ and some fragments with respect to Kripke frames (see Theorem \ref{Thm: distinct fragments} for their distinction).

\ThmIncompleteness*

\begin{proof}
We will prove the theorem for $\mathsf{QFS}$. For the sake of contradiction, suppose a class $\mathcal{C}$ of Kripke frames exists such that $\QFS$ is sound and complete w.r.t. $\mathcal{C}$. Take the 
principle $\varphi:=\mathrm{Lin} \vee \mathrm{OEP}$ where
\[
\mathrm{Lin}:=(C \to D) \vee (D \to C) \quad \mathrm{OEP}:= \exists x A(x) \to \forall x A(x), 
\] 
for any formulas $C$, $D$, and $A(x)$. $\mathrm{Lin}$ is short for Linearity and $\mathrm{OEP}$ for One Element Principle. We will prove that 
all instances of $\phi$ hold (on all frames) in $\mathcal{C}$, but there exists an instance of $\phi$ not valid in $\QFS$.

1) First, we prove $\mathcal{C} \vDash \phi$. By assumption, $\QFS$ is sound and complete w.r.t. $\mathcal{C}$. Therefore, for any $F \in \mathcal{C}$, we have:
\[
F \vDash \CD \qquad F \vDash \ED \qquad F \vDash \SW.
\]
By Lemma \ref{lem: constant domain} and Corollary \ref{Cor: linear or singleton}, $F$ is constant domain, with the domain $\calD$, and either $|\calD|=1$ or $|\calD|>1$ and $F$ is linear. If $|\calD|=1$, then clearly $F \vDash \mathrm{OEP}$. So, suppose $F$ is linear. By Definition \ref{Def: linear}, for any worlds $w,w' \in W$ we have $w \preccurlyeq w'$ or $w' \preccurlyeq w$. Suppose $F \not \vDash \mathrm{Lin}$. Hence, there is a model $\calM$ based on $F$ and a world $w$ such that an instance of $\mathrm{Lin}$ is not valid there
\[
\calM, w \nVdash C \to D \qquad \calM, w \nVdash D \to C.
\]
Thus, there are worlds $v_1 \succcurlyeq w$ and $v_2 \succcurlyeq w$ such that 
\[
\calM, v_1 \Vdash C \quad \calM, v_1 \nVdash D \quad \calM, v_2 \Vdash D \quad \calM, v_2 \nVdash C
\]
However, as $F$ is linear, we have either $v_1 \preccurlyeq v_2$ or $v_2 \preccurlyeq v_1$, which is a contradiction. Therefore $\calc \vDash \phi$.

2) There is an instance of $\phi$ not valid in $\QFS$. To show this, we present a model $\calM$ for $\QFS$ that does not validate that instance. Consider $\calM$ with a language containing only two nullary predicates $P, Q$, and a unary predicate $R(x)$:

\small\[
  \xymatrix{
 w_2 \forc R(a), R(b), P & &   w_3 \forc  R(a), R(b), Q \\
  &\hspace{-25pt}  w_1 \forc R(a) \ar[lu] \ar[ru] & 
  } 
\]
\normalsize where $D_{w_1}=D_{w_2}=D_{w_3}=\{a,b\}$.
First, note that as 
\[
\calM, w_2 \Vdash P \quad \calM, w_2 \nVdash Q \quad \calM, w_3 \Vdash Q \quad \calM, w_3 \nVdash P 
\]
we have
\[
\calM, w_1 \nVdash (P \to Q) \vee (Q \to P).
\]
Moreover, as $\calM, w_1 \Vdash R(a)$ and $\calM, w_1 \nVdash R(b)$ we have
\[
\calM, w_1 \nVdash \exists x R(x) \to \forall x R(x).
\]
Hence, $\calM$ is not a model of the following instance of $\phi$:
\[
\calM, w_1 \nVdash ((P \to Q) \vee (Q \to P)) \vee (\exists x R(x) \to \forall x R(x)).
\]
Now, we want to show that $\mathcal{M}$ is a model of $\mathsf{QFS}$. 

\textbf{Claim.} Take the model $\mathcal{M}$ as above. For any formula $\psi(x)$
\[
\text{either} \quad \calM, w_2 \Vdash \forall x \psi(x)
\quad \text{or} \quad  \calM, w_2 \Vdash \forall x \neg \psi(x)
\]
and
\[
\text{either} \quad \calM, w_3 \Vdash \forall x \psi(x)
\quad \text{or} \quad  \calM, w_3 \Vdash \forall x \neg \psi(x).
\] 
In other words, in any of the final worlds $w_2$ or $w_3$, either both $\psi(a)$ and $\psi(b)$ are valid, or both of them are invalid.

Suppose the claim holds. We prove that $\calM, w_1 \forc \QFS$. 

\item[$\bullet$] As $F$ is constant domain, $\calM, w_1 \Vdash \CD$.

\item[$\bullet$] To show $\calM, w_1 \Vdash \ED$, for the sake of contradiction assume
\[
\calM, w_1 \Vdash B \to \exists x A(x) \qquad \calM, w_1 \nVdash \exists x (B \to A(x)) \tag{$*$}
\]
for formulas $A(x)$ and $B$ where $x$ is not free in $B$. Therefore,
\[
\calM, w_1 \nVdash B \to A(a)  \qquad \calM, w_1 \nVdash B \to A(b). 
\]
Thus, $\exists i, j \in \{1,2,3\}$ such that
\[
\calM, w_i \Vdash B  \quad \calM, w_i \nVdash A(a) \quad \calM, w_j \Vdash B  \quad \calM, w_j \nVdash A(b). 
\]
The following hold for $i$ and $j$:
\begin{itemize}
\item[$\blacktriangleright$]
$i \neq j$: because if $i=j$, then $\calM, w_i \Vdash B$, $\calM, w_i \nVdash A(a)$, and $\calM, w_i \nVdash A(b)$. As $w_i \succcurlyeq w_1$ this means $\calM, w_1 \nVdash B \to \exists x A(x)$, a contradiction with ($*$).

\item[$\blacktriangleright$]
$i \neq 1$: because if $i=1$, then $\calM, w_1 \Vdash B$ and $\calM, w_1 \nVdash A(a)$. By ($*$), $\calM, w_1 \Vdash A(b)$. However, as $w_j \succcurlyeq w_1$ this means $\calM, w_j \Vdash A(b)$, a contradiction.

\item[$\blacktriangleright$]
$j \neq 1$: similar as the above case.
\end{itemize}
Therefore, w.l.o.g. we can assume $i=2$ and $j=3$. Thus,
\[
\calM, w_2 \Vdash B  \quad \calM, w_2 \nVdash A(a).
\]
However, as $\calM, w_1 \Vdash B \to \exists x A(x)$ we have $\calM, w_2 \Vdash A(b)$, which is a contradiction with the claim.

\item[$\bullet$] For the case $\calM, w_1 \Vdash \SW$, see Appendix.

These two points together prove that $\mathsf{QFS}$ is frame-incomplete. The proofs for the fragments are similar. For the proof of the claim, see Appendix.
\end{proof}


\section{Fragments of $\QFS$}

In this subsections, we investigate and compare fragments of $\QFS$. Let us start with the propositional logic underlying $\QFS$. Recall the \emph{propositional logic} of a first-order logic $L$, denoted by $\mathrm{PL}(L)$, as the set of all propositional formulas $\phi$ such that for any first-order substitution $\sigma$ we have $L \vdash \sigma(\phi)$. The following theorem states that quantifier shift formulas lack propositional content, so adding them to $\mathsf{IQC}$ does not alter the intuitionistic propositional logic base.

\begin{theorem}
$\mathrm{PL}(\mathsf{QFS}) = \mathsf{IPC}$.
\end{theorem}
\begin{proof}[Proof sketch] Note that frames of arbitrary form, with the condition that the domain is singleton, are included in the frames for $\QFS$. If we forget about the domain, we get frames for propositional logic $\ipc$. 
\end{proof}

As expected for formulas in prenex normal form we have:

\begin{restatable}{lemma}{LemPrenex}\label{lem: prenex}
Let $F$ be an intermediate logic such that $\QFS \subseteq L$. For any formula $A$, there exists a formula $B$ in the prenex normal form such that $L \vdash (A \to B) \wedge (B \to A)$. 
\end{restatable}
\begin{proof}
    See Appendix.
\end{proof}

The following theorem establishes that $\mathsf{QFS}$, $\mathsf{IQC}+ \{\mathsf{CD} , \mathsf{SW}\},$ and $\mathsf{IQC}+ \{\mathsf{CD} , \mathsf{ED}\}$ are all distinct. 

\begin{restatable}{theorem}{ThmFragments} \label{Thm: distinct fragments}
The following hold:
\begin{enumerate}
\item
$\IQC+\{\mathsf{CD}, \mathsf{SW}\} \nvdash \mathsf{ED}$. 
\item
$\IQC + \{\mathsf{CD}, \mathsf{ED}\} \nvdash \mathsf{SW}$. 
\item \label{SW CD}
$\mathsf{IQC}+\SW \vdash \CD$
 \item \label{ED CD}
 $\mathsf{IQC}+\ED \nvdash \CD$
\end{enumerate}
\end{restatable}  

\begin{proof}
For 1 and 2, use similar models as in the proof of Lemma \ref{lem: WF}. The proof is similar to the proof of Lemma \ref{lem: WF validates ED}. For more details, see Appendix.
\end{proof}

\section{Conclusion and Future Work}
In this article, the intermediate logics that admit standard
Skolemization have been characterized:
those that contain all quantifier shift principles. This does not
extend, however, to prenex fragments: The prenex fragments of 
intuitionistic logic, G\"odel logics (the logics of countable linear Kripke
frames with constant domains), etc. admit standard Skolemization (\cite{Mints66,Mints72}). It should
be emphasized that contrary to classical logic, deSkolemization is not
possible in general if identity axioms for Skolem functions are used
\cite{Mints2000}. This question remains open for
the intermediate logic with quantifier shift principles. Additional fragments and
variants could be considered similar to the the standard Skolemization
of the existential fragment of intuitionistic logic with the existence predicate \cite{baaz2006skolemization}.

For $\QFS$ in itself the development
of suitable semantics is desirable.
By Example \ref{unsound-ex},  $\QFS$ lacks the
existence property, but this is unclear for the disjunction property as
the underlying propositional logic is intuitionistic logic; therefore, the
disjunction property holds for $A\vee B$, where $A$ and $B$ are propositional. An additional
proof theoretic property would be the existence of Herbrand expansions
implied by the addition of quantifier shift principles.

The arguments of this article extend to many other classes of logics:
Firstly, the direction, deSkolemization $\Rightarrow$ quantifier shifts, is very
general and can be argued for almost any logic. Here the determination
of suitable minimal conditions would be useful. Secondly, the argument
in the other direction depends on the existence of an underlying
analytic calculus with a corresponding $++$ variant, e.g., for normal first-order modal logics where the sequent calculus for $K$ could be used. Here,
the condition is the existence of the shift of quantifiers through
modalities.

As Example \ref{SL} shows, it can be worthwhile to depart from standard
Skolemization, cf. the definable Skolem functions in Real Closed Fields.
Of course, there are limits, as the error in Hilbert-Bernays \cite{hilbert} on
Disparate Functions demonstrates.

 From a more foundational perspective, the arguments of this article can
be used to demonstrate the conservativity of critical epsilon formulas
$A(t)\rightarrow A(\varepsilon_xA(x))$ beyond the first epsilon theorem, which holds in very few cases, cf. Example \ref{unsound-ex}. The epsilon terms 
can
be considered as compositions of Skolem functions, so the layer-wise
deSkolimization leads to the conservativity of epsilon terms
when the deSkolemization is possible.

To put together, the aim of this article is to initiate a program for the
comprehensive analysis of Skolemization properties in first-order logics.

\bibliographystyle{plain}
\bibliography{QFS}

\newpage
\section{Appendix}
This appendix primarily contains (parts of) proofs omitted from the main text due to space constraints. For the reader's convenience, we restate the corresponding lemmas and theorems. While most proofs are straightforward for a keen reader, some involve technical challenges that rely on specific tricks (e.g., item 4 in Theorem \ref{Thm: distinct fragments}). We believe their omission does not detract from the overall quality or completeness of the paper.


\noindent \textbf{Justification for Definition \ref{Def: suitable}:}

\noindent Let us justify the choice of a suitable quantifier. Violation of each condition in Definition \ref{Def: suitable quantifier} leads to an undesirable proof:

\begin{enumerate}
\item
(substitutability) \ \ \ \ \qquad
\AxiomC{$A(a) \Rightarrow A(a)$}
\UnaryInfC{$A(a) \Rightarrow \forall x A(x)$}
\DisplayProof
\item
(side-variable condition)
\begin{center}
\AxiomC{$A(a,b) \Rightarrow A(a,b)$}
\RightLabel{$*_1$}
\UnaryInfC{$\exists y A(a,y) \Rightarrow A(a, b)$}
\RightLabel{$*_2$}
\UnaryInfC{$\exists y A(a,y) \Rightarrow \forall x A(x, b)$}
\UnaryInfC{$\forall x \exists y A(x,y) \Rightarrow \forall x A(x, b)$}
\UnaryInfC{$\forall x \exists y A(x,y) \Rightarrow \exists y \forall x A(x, y)$}
\DisplayProof
\end{center}
Note that in $*_1$ and $*_2$ we have $b <_\pi a$ and $a <_\pi b$, respectively. 
\item
(very weak regularity)
\begin{center}
\AxiomC{$A(a) \Rightarrow A(a)$}
\UnaryInfC{$\exists x A(x) \Rightarrow A(a)$}
\AxiomC{$A(a) \Rightarrow A(a)$}
\UnaryInfC{$A(a) \Rightarrow \forall x A(x)$}
\BinaryInfC{$\exists x A(x) \Rightarrow \forall x A(x)$}
\DisplayProof
\end{center}
\end{enumerate}

\ThmQFSLJpp*
\begin{proof}
We prove each direction.

$(\Leftarrow)$ Let $\mathbf{QFS} \vdash S$. If $\mathbf{LJ} \vdash S$ then $\mathbf{LJ}^{++} \vdash S$. Moreover, the axioms $\mathsf{CD}$, $\mathsf{ED}$, and $\mathsf{SW}$ are provable in $\mathbf{LJ}^{++}$:
\begin{center}
\begin{tabular}{c c}
\hspace{-20pt}
$\mathbf{LJ}^{++} \vdash \SW$
&
\hspace{20pt}
\AxiomC{$A(a) \Rightarrow A(a)$}
\RightLabel{$*$}
\UnaryInfC{$A(a) \Rightarrow \forall x A(x)$}
\AxiomC{$B \Rightarrow B$}
\BinaryInfC{$A(a), \forall x A(x) \to B \Rightarrow B$}
\UnaryInfC{$\forall x A(x) \to B \Rightarrow A(a) \to B$}
\UnaryInfC{$\forall x A(x) \to B \Rightarrow \exists x (A(x) \to B)$}
\DisplayProof
\end{tabular}
\end{center}

\begin{center}
\begin{tabular}{c c}
\hspace{-20pt}
$\mathbf{LJ}^{++} \vdash \ED$
&
\hspace{20pt}
\AxiomC{$A(a) \Rightarrow A(a)$}
\RightLabel{$*$}
\UnaryInfC{$\exists x A(x) \Rightarrow A(a)$}
\AxiomC{$B \Rightarrow B$}
\BinaryInfC{$B, B \to \exists x A(x) \Rightarrow A(a)$}
\UnaryInfC{$B \to \exists x A(x) \Rightarrow B \to A(a)$}
\UnaryInfC{$B \to \exists x A(x) \Rightarrow \exists x (B \to A(x))$}
\DisplayProof
\end{tabular}
\end{center}

\begin{center}
\begin{tabular}{c c}
\hspace{-10pt}
$\mathbf{LJ}^{++} \vdash \CD$
&
\small
\AxiomC{$A(a) \Rightarrow A(a)$}
\RightLabel{$*$}
\UnaryInfC{$ A(a) \Rightarrow \forall x A(x)$}
\UnaryInfC{$A(a) \Rightarrow \forall x A(x) \vee B$}
\AxiomC{$B \Rightarrow \forall x A(x) \vee B$}
\BinaryInfC{$A(a) \vee B \Rightarrow \forall x A(x) \vee B$}
\UnaryInfC{$\forall x (A(x) \vee B) \Rightarrow \forall x A(x) \vee B$}
\DisplayProof
\end{tabular}
\end{center}
Note that the rules marked by $(*)$ make the proofs unsound for $\mathbf{LJ}$.

$(\Rightarrow)$ Let $\mathbf{LJ}^{++} \vdash^{\pi} S$, where $S=(\Sigma \Rightarrow \Lambda)$. Change $\pi$ to $\pi'$ by replacing each unsound universal and existential quantifier inference as follows:
\small \begin{center}
\begin{tabular}{c}
\AxiomC{$\Gamma \Rightarrow A(a)$}
\UnaryInfC{$\Gamma \Rightarrow \forall x A(x)$}
\DisplayProof
$\rightsquigarrow$
\AxiomC{$\Gamma \Rightarrow A(a)$}
\AxiomC{$\forall x A(x) \Rightarrow \forall x A(x)$}
\BinaryInfC{$\Gamma, A(a) \to \forall x A(x) \Rightarrow \forall x A(x)$}
\DisplayProof
\end{tabular}
\end{center}

\small \begin{center}
\begin{tabular}{c}
\AxiomC{$\Gamma, B(b) \Rightarrow \Delta$}
\UnaryInfC{$\Gamma, \exists x B(x) \Rightarrow \Delta$}
\DisplayProof
$\rightsquigarrow$
\AxiomC{$\exists x B(x) \Rightarrow \exists x B(x)$}
\AxiomC{$B(b), \Gamma \Rightarrow \Delta$}
\BinaryInfC{$\Gamma, \exists x B(x), \exists x B(x) \to B(b) \Rightarrow \Delta$}
\DisplayProof
\end{tabular}
\end{center}
\normalsize Denote $A_i(a_i) \rightarrow \forall x A_i(x)$ by $\alpha_i(a_i)$ and $\exists x B_j(x) \rightarrow B_j(b_j)$ by $\beta_j(b_j)$.
Therefore, $\pi'$ is a proof in $\mathbf{LJ}$ such that
\[
\mathbf{LJ} \vdash^{\pi'} \{\alpha_i\}_{i \in I}, \{\beta_j\}_{j \in J}, \Sigma \Rightarrow \Lambda, \tag{1}
\]
for (possibly empty) sets $I$ and $J$. Now, we claim that we can use the rule $(\exists L)$ as many times as needed to get 
\[
\mathbf{LJ} \vdash \{\exists y_i \alpha_i (y_i)\}_{i \in I}, \{\exists z_j \beta_j(z_j))\}_{j \in J}, \Sigma \Rightarrow \Lambda. \tag{2}
\]
Suppose the claim holds. To get a proof of $S$ in $\mathbf{QFS}$, we remove the additional formulas in $(2)$. 
For each $i \in I$ and $j \in J$, we have 
\small\[
\mathbf{QFS} \vdash \ \Rightarrow (\forall y_i A_i(y_i) \to \forall x A_i (x)) \to \exists y_i (A_i(y_i)\to \forall x A_i (x)),
\]
\normalsize since the succedent of the sequent is an instance of the axiom $\SW$ and hence provable in $\mathbf{QFS}$. Similarly,
\small \[
\mathbf{QFS} \vdash \ \Rightarrow (\exists x B_j(x) \to \exists z_j B_j(z_j)) \to \exists z_j(\exists x B_j(x) \rightarrow B_j(z_j)),
\]
\normalsize as the succedent is an instance of the axiom $\ED$. Therefore, 
\[
\mathbf{QFS} \vdash \ \Rightarrow \exists y_i (A_i(y_i)\to \forall x A_i (x)),
\]
\[\mathbf{QFS} \vdash \ \Rightarrow \exists z_j(\exists x B_j(x) \rightarrow B_j(z_j)).
\]
By the cut rule on the above sequents and the sequent $(2)$, we get $\mathbf{QFS} \vdash S$.

Now, we prove the claim. To make sure that the existential quantifier inferences can be carried out in $\mathbf{LJ}$, (i.e., deriving $(2)$ from $(1)$), we use the fact that the initial proof $\pi$ only contained \textit{suitable} quantifier inferences. We will explain how each condition is used. 

By the substitutability, we know that as each $a_i$ and $b_j$ is a characteristic variable, they do not appear in the conclusion of $\pi$; namely, they do not appear in $\Sigma$ and $\Lambda$. 

By very weak regularity, if $a$ is the characteristic variable of two strong-quantifier inferences in $\pi$, then they must have the same principal formulas. Therefore, we cannot have both
\[
A(a) \to \forall x A(x) \qquad \text{and} \qquad \exists x B(x) \to B(a)
\]
or both
\[
A(a) \to \forall x A(x) \qquad \text{and} \qquad C(a) \to \forall x C(x)
\]
present in $(2)$.

Finally, by the side-preserving condition, we cannot have
\[
A(a,b) \to \forall x A(x, b) \qquad \text{and} \qquad \exists x B(a, x) \to B(a,b)
\]
both present in $(2)$, as this would require $a <_\pi b$ and $b <_\pi a$ which is in contradiction with $<_\pi$ being acyclic. Thus, the claim holds.
\end{proof}

\begin{example}\label{Ex: deSkolem}
Consider the Skolemization of proof of a target end-sequent $\forall x\exists yA(x,y)\Rightarrow \forall x\exists yA(x,y)$ 
\begin{center}
\AxiomC{$A(c,f(c)) \Rightarrow  A(c,f(c))$}
\UnaryInfC{$A(c,f(c))  \Rightarrow  \exists y A(c,y) $}
\UnaryInfC{$\forall x A(x,f(x)) \Rightarrow \exists y A(c,y)$}
\DisplayProof
\end{center}
Now we deSkolemize the above proof. $a_c$ and $a_{f(c)}$ are new variables for Skolem terms $c$ and $f(c)$.
\begin{center}
\AxiomC{$A(a_c,a_{f(c)}) \Rightarrow  A(a_c,a_{f(c)})$}
\UnaryInfC{$A(a_c,a_{f(c)})  \Rightarrow  \exists y A(a_c,y) $}
\UnaryInfC{$A(a_c,a_{f(c)}) \Rightarrow \forall x\exists y A(x,y)$}
\UnaryInfC{$  \exists y A(a_c,y) \Rightarrow \forall x\exists y A(x,y)$}
\UnaryInfC{$  \forall x\exists y A(x,y) \Rightarrow \forall x\exists y A(x,y)$}
\DisplayProof
\end{center}
The order of the variables is $a_c<a_{f(c)}$ which is also the side variable condition and acyclic.\\
The following quantifier inferences are added in the above proof
\begin{align*}
    A(a_c,a_{f(c)}\Rightarrow \exists y A(a_c,y) \\
    \exists y A(a_c,y) \Rightarrow \forall x\exists y A(x,y)
\end{align*}
\end{example}

\LemConstantDomain*
\begin{proof}
We prove each case by contraposition, i.e., in each case, we assume $F$ is not constant domain. Then, we construct a model and find a world and an instance of the axiom that is not valid in that world. If $F$ is not constant domain, then
\[
\exists w \neq w' \ w \preccurlyeq w' \ \exists a \in D_{w'} - D_w.
\]
Let $Q$ be a nullary predicate and $P$ a unary predicate.

\item[1)] Similar to \cite[Lemma 14.4.]{Mints}. Construct the model $\calM$ based on $F$ as follows: 
\begin{align*}
\calM , u & \Vdash Q & \text{iff} \ \  u \succcurlyeq w, u \neq w \\
\forall d \in D_w \ \ \calM , u & \Vdash P(d) & \text{iff} \ \ u \succcurlyeq w \\ 
\forall a \in D_{w'}-D_w \ \ \calM , u & \nVdash P(a) & \text{iff} \ \ u \succcurlyeq w'
\end{align*}
depicted as: 
\[
  \xymatrix{
    \hspace{-25pt} \forall a \in D_w' - D_w \ w' \nvDash P(a) \ w' \forc Q  
 \\
   \vdots \ar[u]   
  \\
  \hspace{-25pt} \forall d \in D_w \ w \forc P(d) \ar[u] 
  } 
\]
Therefore, as 
\[
\forall d \in D_w \ \calM, w \Vdash P(d) \quad \text{and} \quad \forall u \succcurlyeq w \ u \neq w \ \calM, u \Vdash Q,
\]
we get $\calM, w \forc \forall x (P(x) \vee Q)$. However, as
\[
\calM, w \nVdash Q \quad \text{and} \quad D_{w'}-D_w \neq \emptyset.
\]
We have $\calM, w \nVdash \forall x P(x) \vee Q$.
Hence,
\[
\calM, w \nVdash \forall x (P(x) \vee Q)  \to \forall x P(x) \vee Q.
\]

\item[2)] Construct a model $\calM$ based on $F$ as follows: 
\[
\forall a \in D_{w'}-D_w \ \ \calM , u \Vdash P(a), Q \ \ \text{iff} \ \  u \succcurlyeq w',
\]
depicted as:
\[
  \xymatrix{
  &  \hspace{-45pt} \forall a \in D_w' - D_w \ \ w' \forc P(a), Q  & \\
 &  \vdots \ar[u] &  & 
 \\
  & w \ar[u] &  &  
  } 
\]
Therefore, in every world above $w$ if $Q$ is valid, then $P(a)$ is valid for all $a \in D_{w'}-D_w$. However, such $a$ are not in the domain of $w$. 
Hence, we have 
\[
\calM, w \Vdash Q \to \exists x P(x) \qquad \calM, w \nVdash \exists x (Q \to P(x)).
\]

\item[3)] The same model as in (1) works. As $D_{w'} - D_w \neq \emptyset$
\[
\calM, w \nVdash \forall x P(x).
\]
Thus, $\calM, w \forc \forall x P(x) \to Q$. However, $\calM, w \nVdash Q$ and the elements in $D_{w'}-D_w$ obviously do not belong to $D_w$. Hence, 
\[
\calM, w \Vdash \forall x P(x) \to Q \qquad \calM, w \nVdash \exists x (P(x) \to Q).
\]
\end{proof}

\LemFDS*
\begin{proof}
If $F \vDash \ED$ or $F \vDash \SW$, then by Lemma \ref{lem: constant domain} it is constant domain. Call this domain $\calD$. We prove each case by contraposition. Suppose $F$ does not satisfy $\FDS$. Therefore, 
\[
\exists v_1,v_2,w \in W \quad v_1 \neq v_2 \ v_1 \succcurlyeq w \ v_2 \succcurlyeq w \quad \text{and} \quad |\calD|>1.
\]
As $|\calD|>1$, there are distinct elements $a,b \in \calD$. Let $Q$ be a nullary and $P$ a unary predicate. 

\item[1)] We want to prove $F \not \vDash \ED$. Construct the model $\calM$ as:
\begin{align*}
\calM , u \Vdash P(a), Q & & \text{iff} & &  u \succcurlyeq v_1, \\
\calM , u \Vdash P(b), Q & & \text{iff} & &  u \succcurlyeq v_2,
\end{align*}
depicted as
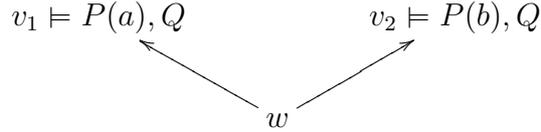
\[
  \xymatrix{
  v_1 \vDash P(a), Q & &   v_2 \vDash P(b),Q \\
  &  
   w  \ar[lu] \ar[ru] & 
  } 
\]
Then, we get
\[
\calM, w \Vdash Q \to \exists x P(x) \qquad \calM, w \nVdash \exists x (Q \to P(x)).
\]

\item[2)] To prove $F \not \vDash \SW$, construct the model $\calM$ as:
\begin{align*}
\calM , u \Vdash P(a) & & \text{iff} & & u \succcurlyeq v_1, \\
\calM , u \Vdash P(b) & &  \text{iff} & & u \succcurlyeq v_2,
\end{align*}
depicted as:
\[
  \xymatrix{
  v_1 \vDash P(a) & &  
  v_2 \vDash P(b)  \\
  & w  \ar[lu] \ar[ru] & 
  } 
\]
Then,
\[
\calM, w \Vdash \forall x P(x) \to Q \qquad \calM, w \nVdash \exists x (P(x) \to Q).
\]
\end{proof}

\LemWF*
\begin{proof} If $F \vDash \ED$ or $F \vDash \SW$, then by Lemma \ref{lem: constant domain} it is constant domain. By the assumption, the domain is countably infinite. Call it $\calD=\{a_0, a_1, \dots\}$. Thus, by Corollary \ref{Cor: linear or singleton}, $F$ is linear. We prove each case by contraposition.  Let $Q$ be a nullary predicate and $P$ a unary one.
\item[1)] 
Suppose $F$ does not satisfy $\WF$. Thus, there exists $w \in W$ where $\preccurlyeq [w]$ contains an infinite descending chain $v_0 \succcurlyeq v_1 \succcurlyeq v_2 \succcurlyeq \dots$. Call this chain $C$. Construct the model $\calM$:
\begin{align*}
\calM , u & \Vdash Q \quad \text{iff} \ \  u \succcurlyeq v_0 \ \text{or} \ u \in C \\
\forall d \in \calD \ \ \calM , v_0 & \Vdash P(d) \\ 
\forall d \in \calD \ \ \calM , w & \nVdash P(d)   \\ 
\forall d \in \calD - \{a_0, \dots, a_{i-1}\} \ \ \calM, u & \Vdash P(d) \qquad i \geq 1, \forall u \succcurlyeq v_i
\end{align*}

depicted as:
\[
  \xymatrix{ 
  \vdots  \\
            \hspace{-15pt}   \forall d \in \mathcal{D} \quad   v_0 \forc Q, P(d) \ar[u]   \\
           \hspace{-45pt}   \forall d \in \mathcal{D} - \{a_0\} \quad   v_1 \forc Q, P(d)  \ar[u]    \\
         \hspace{-58pt}   \forall d \in \mathcal{D} - \{a_0, a_1\} \quad   v_2 \forc Q, P(d)  \ar[u]    \\
           \vdots \ar[u] \\
           w \ar[u]
           }
\]
For any world $u \succcurlyeq w$ there are three cases:

$\bullet$
$u \in C$. Thus, there exists $i \geq 0$ such that $u=v_i$. Hence, $\calM, v_i \forc Q, P(a_i)$.

$\bullet$
$u \succcurlyeq v_0$. Thus, $\calM, u \forc Q, P(a_0)$.

$\bullet$
Otherwise. Hence, $\calM, u \nVdash Q$.

\noindent Therefore, we get $\calM, w \Vdash Q \to \exists x P(x)$. However, as $\calM, w \nVdash Q \to P(a_i)$ for any $i \geq 0$, we get $\calM, w \nVdash \exists x (Q \to P(x))$. Thus,
\[
\calM, w \nVdash (Q \to \exists x P(x)) \to \exists x (Q \to P(x)).
\]

\item[2)] Similar to the previous case. Suppose $F$ does not satisfy $\cWF$. Thus, there exists $w \in W$ where $\preccurlyeq [w]$ contains an infinite ascending chain $v_0 \preccurlyeq v_1 \preccurlyeq v_2 \preccurlyeq \dots$. Call this chain $C$. Construct the model $\calM$:
\begin{align*}
\calM , u & \nVdash Q & \forall u \in W \\
\forall d \in \{a_0, \dots, a_{i}\} \ \ \calM, v_i & \Vdash P(d) & i \geq 0, v_i \in C
\end{align*}
depicted as:
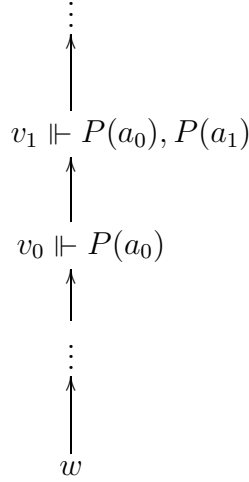
\[
  \xymatrix{ 
             \vdots \\
           \hspace{45pt}  v_1 \forc P(a_0), P(a_1) \ar[u] \\ 
      \hspace{15pt}   v_0 \forc P(a_0) \ar[u] \\
           \vdots \ar[u] \\
         w  \ar[u]
           }
\]
Therefore, $\calM, w \nVdash \forall x P(x)$. Thus,
\[
\calM, w \Vdash \forall x P(x) \to Q.
\]
However, as $\calM, v_i \forc P(a_i)$  for every $a_i \in \calD$, we get $\calM, w \nVdash \exists x (P(x) \to Q)$. Hence,
\[
\calM, w \nVdash (\forall x P(x) \to Q) \to \exists x (P(x) \to Q).
\]
\end{proof}

\LemConsDomCD*
\begin{proof} Similar to \cite[Lemma 14.4.]{Mints}. See Appendix. Let $F$ be constant domain. If $F \nvDash \CD$, then there is a model $\calM$ and world $w$ such that 
\[
\calM, w \Vdash \forall x (P(x) \vee Q) \qquad \calM, w \nVdash \forall x P(x) \vee Q,
\]
for a unary predicate $P$ and a nullary one $Q$, where $x$ is not free in $Q$.
The right statement means that 
\[
\calM, w \nVdash \forall x P(x)  \qquad \calM, w \nVdash  Q.
\]
Together with the left statement this implies that there exists a world $v \succcurlyeq w$ and an element $d \in D_v$ such that $\calM, v \nVdash P(d)$. However, since the frame is constant domain, $d \in D_w$, which means that $\calM, w \nVdash P(d)$. This is in contradiction with the fact that $\calM, w \Vdash \forall x (P(x) \vee Q)$ and $w \succcurlyeq w$ and $\calM, w \nVdash Q$.
\end{proof}

\LemSingleton*
\begin{proof}
Let $F=(W, \preceq, D)$, where $|D_w|=1$ for every $w \in W$. Call this element $a$.

\item[$\bullet$]
As $F$ is constant domain, by Lemma \ref{lem: constant domain}, we get $F \vDash \CD$.

\item[$\bullet$]
To prove $F \vDash \SW$, suppose 
\[
F \vDash \forall x A(x) \to B,
\]
for any formulas $A$ and $B$, where $x$ is not free in $B$. This means that for any model $\calM$ based on $F$ and any world $w$, we have $\calM, w \Vdash \forall x A(x) \to B$, which implies 
\[
\forall w' \succcurlyeq w \quad \text{if} \quad \calM, w' \Vdash \forall x A(x) \quad \text{then} \quad 
 \calM, w' \Vdash B.
\]
However, as the domain is $\{a\}$, we get
\[
\forall w' \succcurlyeq w \quad \text{if} \quad \calM, w' \Vdash  A(a) \quad \text{then} \quad 
 \calM, w' \Vdash B,
\]
which means $\calM, w \Vdash A(a) \to B$. Therefore, $\calM, w \Vdash \exists x A(x) \to B$. Thus, $F \vDash \SW$.
\item[$\bullet$]
To prove $F \vDash \ED$, suppose $F \vDash B \to \exists x A(x)$, for any formulas $A$ and $B$, where $x$ is not free in $B$. This means for any model $\calM$ based on $F$ and any world $w$, we have $\calM, w \Vdash B \to \exists x A(x)$. As the domain of each world is $\{a\}$, we get $\calM, w \Vdash B \to A(a)$. Thus, $\calM, w \Vdash \exists x (B \to A(x))$ and hence $F \vDash \ED$.
\end{proof}

\LemFiniteDom*
\begin{proof}
Let $F$ be a linear constant domain frame, where the domain is $\calD=\{a_1, a_2, \dots, a_n\}$, for some natural number $n$. By Lemma \ref{lem: cons dom validates CD}, $F \vDash \CD$. Now, suppose $F \nvDash \ED$. Hence, there exists a model $\calM$ and a world $w \in W$ such that 
\[
\calM, w \Vdash Q \to \exists x P(x) \qquad 
\calM, w \nVdash \exists x (Q \to P(x)) \tag{1}
\]
for formulas $P$ and $Q$, where $x$ is not free in $Q$. The right statement implies $\calM, w \nVdash Q \to P(a_i)$, for each $1 \leq i \leq n$. Thus,
\[
\exists v_i \geq w \quad \calM, v_i \Vdash Q \quad \calM, v_i \nVdash P(a_i).
\]
As $F$ is linear, the minimum of the set $\{v_1, v_2, \dots, v_n\}$ exists. W.l.o.g. suppose the minimum is $v_1$. Hence, 
\[
\calM, v_1 \Vdash Q \qquad \forall 1 \leq i \leq n \ \calM, v_1 \nvDash P(a_i),
\]
which implies
\[
\calM, v_1 \nVdash Q \to \exists x P(x),
\]
which is a contradiction with (1). Therefore, $F \vDash \ED$.

The case for $\SW$ is similar with the difference that this time we reason by taking the maximum of the worlds. Suppose $F \nvDash\SW$. Thus, there is a model $\calM$ and world $w$
\[
\calM , w \Vdash \forall x P(x) \to Q \qquad \calM, w \nVdash \exists x (P(x) \to Q) \tag{2}
\]
for formulas $P$ and $Q$, where $x$ is not free in $Q$. The right statement implies $\calM, w \nVdash P(a_i) \to Q$, for each $1 \leq i \leq n$. Hence,
\[
\exists v_i \geq w \quad \calM, v_i \Vdash P(a_i) \quad \calM, v_i \nVdash Q.
\]
W.l.o.g. suppose $v_n$ is the maximum of the set $\{v_1, \dots, v_n\}$. Thus,
\[
 \forall 1 \leq i \leq n \quad \calM, v_n \Vdash P(a_i) \qquad \calM, v_n \nVdash Q,
\]
which is a contradiction with (2). Therefore, $F \vDash \SW$.
\end{proof}

\lemWFvalidatesED*

\begin{proof}
Let $F$ be a linear constant domain frame. Denote the domain by $\calD$.

1) Let $F$ satisfy $\WF$. We want to prove $F \vDash \ED$, i.e.,
\[
F \vDash (B \to \exists x A(x)) \to \exists x (B \to A(x)),
\]
for any formulas $A(x)$ and $B$, where $x$ is not free in $B$. Thus, for any model $\calM$ based on $F$ and any $w \in W$, we will show
\[
\text{if} \quad \calM, w \Vdash B \to \exists x A(x) \quad \text{then} \quad \calM, w \Vdash \exists x (B \to A(x)).
\]
Suppose $\calM, w \Vdash B \to \exists x A(x)$, i.e., for any $v \succcurlyeq w$
\[
\text{if} \quad \calM, v \Vdash B \quad \text{then} \quad \calM, v \Vdash \exists x A(x).
\]
There are two cases:

$\bullet$ $\calM, v \nVdash B$ for all $v \succcurlyeq w$. Take $d \in \mathcal{D}$ arbitrarily. Hence, $\calM, w \Vdash B \to A(d)$, which means $\calM, w \Vdash \exists x (B \to A(x))$.

$\bullet$ $\calM, v \Vdash B$ for some $v \succcurlyeq w$. As $F$ satisfies $\WF$ the set $\preccurlyeq [w]$ is well-founded. Take $S=\{v \succcurlyeq w \mid \calM, v \Vdash B\}$. Clearly $S \subseteq \preccurlyeq [w]$, hence it has a minimal element $u$. Thus, 
\[
\calM, u \Vdash B \quad \text{hence} \quad \calM, u \Vdash \exists x A(x).
\]
Thus, there exists $d \in \mathcal{D}$ such that $\calM, u \Vdash A(d)$. We will show $\calM, w \Vdash B \to A(d)$. The reason is that for any world $w'$ such that $w \preccurlyeq w' \preccurlyeq u$ we have
\[
\calM, w' \nVdash B \quad \text{hence} \quad \calM, w' \Vdash B \to A(d).
\]
Moreover, for $w' \succcurlyeq u$, as $\calM, u \Vdash B \to A(d)$ we have
\[
\calM, w' \Vdash B \to A(d).
\]
Hence, by definition $\calM, w \Vdash B \to A(d)$. As $F$ is constant domain, $d \in D_w$. Thus, $\calM, w \Vdash \exists x (B \to A(x))$.


2) Similar to the previous case. Let $F$ satisfy $\cWF$. We want to prove $F \vDash \SW$, i.e.,
\[
F \vDash (\forall x A(x) \to B) \to \exists x (A(x) \to B),
\]
for formulas $A(x)$ and $B$, where $x$ is not free in $B$. Thus, for any model $\calM$ based on $F$ and any $w \in W$, we will show
\[
\text{if} \quad \calM, w \Vdash \forall x A(x) \to B \quad \text{then} \quad \calM, w \Vdash \exists x (A(x) \to B).
\]
Suppose $\calM, w \Vdash \forall x A(x) \to B$, i.e., for any $v \succcurlyeq w$
\[
\text{if} \quad \calM, v \Vdash \forall x A(x) \quad \text{then} \quad \calM, v \Vdash B.
\]
There are two cases:

$\bullet$ $\calM, v \nVdash \forall x A(x)$ for all $v \succcurlyeq w$. As $F$ satisfies $\cWF$, the set $\preccurlyeq [w]$ is conversely well-founded and has a maximal element $u$. As $\calM, u \nVdash \forall x A(x)$, there exists $d \in \calD$ such that $\calM, u \nVdash A(d)$. As $F$ is constant domain, $d \in D_w$ and for any $w'$ such that $w \preccurlyeq w' \preccurlyeq u$ we have $\calM, w' \nVdash A(d)$. Therefore, $\calM, w \Vdash A(d) \to B$. Hence, $\calM, w \Vdash \exists x ( A(x) \to B)$.

$\bullet$ $\calM, v \Vdash \forall x A(x)$ for some $v \succcurlyeq w$. As $F$ satisfies $\cWF$ the set $S=\{v \succcurlyeq w \mid \calM, v \nVdash \forall x A(x)\}$ is conversely well-founded. As $S \subseteq \preccurlyeq [w]$, it has a maximal element $u$. Thus, 
\[
\calM, u \nVdash \forall x A(x).
\]
Hence, there is $d \in \mathcal{D}$ such that $\calM, u \nVdash A(d)$. Moreover, 
$\forall w' \succcurlyeq u$, where $w \neq u$,
we have $\calM, w' \Vdash \forall x A(x)$. Hence, $\calM, w' \Vdash B$. 
Now, we show that $\calM, w \Vdash A(d) \to B$. Thus we have to show that for any $w'\succcurlyeq w$: 
\[
 \calM, w' \Vdash A(d) \text{ implies }\calM, w' \Vdash B.
\]
For $w' \preccurlyeq u$, this clearly holds because $w'\not\Vdash A(d)$. For $w' \succ u$, as $\calM, w' \Vdash \forall x A(x)$ and $\calM, w \Vdash \forall x A(x) \to B$, it follows that 
$\calM, w' \Vdash B$. 
Thus we have shown that $\calM, w \Vdash A(d) \to B$. As $F$ is constant domain, $d \in D_w$. Thus, $\calM, w \Vdash \exists x (A(x) \to B)$.

3) Clear by the previous parts. In particular, if $F$ is finite, then it satisfies both $\WF$ and $\cWF$. 
\end{proof}

\ThmIncompleteness*
\begin{proof}
We will prove the theorem for $\mathsf{QFS}$. For the sake of contradiction, suppose a class $\mathcal{C}$ of Kripke frames exists such that $\QFS$ is sound and complete w.r.t. $\mathcal{C}$. Take the 
principle $\varphi:=\mathrm{Lin} \vee \mathrm{OEP}$ where
\[
\mathrm{Lin}:=(C \to D) \vee (D \to C) \quad \mathrm{OEP}:= \exists x A(x) \to \forall x A(x), 
\] 
for any formulas $C$, $D$, and $A(x)$. $\mathrm{Lin}$ is short for Linearity and $\mathrm{OEP}$ for One Element Principle. We will prove that 
all instances of $\phi$ hold (on all frames) in $\mathcal{C}$, but there exists an instance of $\phi$ not valid in $\QFS$.

1) First, we prove $\mathcal{C} \vDash \phi$. By assumption, $\QFS$ is sound and complete w.r.t. $\mathcal{C}$. Therefore, for any $F \in \mathcal{C}$, we have:
\[
F \vDash \CD \qquad F \vDash \ED \qquad F \vDash \SW.
\]
By Lemma \ref{lem: constant domain} and Corollary \ref{Cor: linear or singleton}, $F$ is constant domain, with the domain $\calD$, and either $|\calD|=1$ or $|\calD|>1$ and $F$ is linear. If $|\calD|=1$, then clearly $F \vDash \mathrm{OEP}$. So, suppose $F$ is linear. By Definition \ref{Def: linear}, for any worlds $w,w' \in W$ we have $w \preccurlyeq w'$ or $w' \preccurlyeq w$. Suppose $F \not \vDash \mathrm{Lin}$. Hence, there is a model $\calM$ based on $F$ and a world $w$ such that an instance of $\mathrm{Lin}$ is not valid there
\[
\calM, w \nVdash C \to D \qquad \calM, w \nVdash D \to C.
\]
Thus, there are worlds $v_1 \succcurlyeq w$ and $v_2 \succcurlyeq w$ such that 
\[
\calM, v_1 \Vdash C \quad \calM, v_1 \nVdash D \quad \calM, v_2 \Vdash D \quad \calM, v_2 \nVdash C
\]
However, as $F$ is linear, we have either $v_1 \preccurlyeq v_2$ or $v_2 \preccurlyeq v_1$, which is a contradiction. Therefore, $\calc \vDash \phi$.

2) There is an instance of $\phi$ not valid in $\QFS$. To show this, we present a model $\calM$ for $\QFS$ that does not validate that instance. Consider $\calM$ with a language containing only two nullary predicates $P, Q$, and a unary predicate $R(x)$:
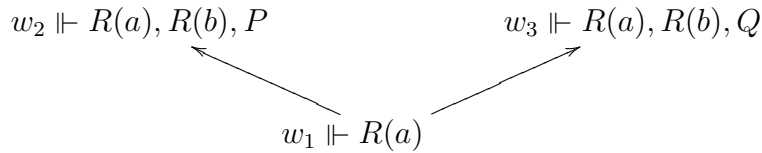
\[
  \xymatrix{
 w_2 \forc R(a), R(b), P & &   w_3 \forc  R(a), R(b), Q \\
  &\hspace{-25pt}  w_1 \forc R(a) \ar[lu] \ar[ru] & 
  } 
\]
where $D_{w_1}=D_{w_2}=D_{w_3}=\{a,b\}$.
First, note that as 
\[
\calM, w_2 \Vdash P \quad \calM, w_2 \nVdash Q \quad \calM, w_3 \Vdash Q \quad \calM, w_3 \nVdash P 
\]
we have
\[
\calM, w_1 \nVdash (P \to Q) \vee (Q \to P).
\]
Moreover, as $\calM, w_1 \Vdash R(a)$ and $\calM, w_1 \nVdash R(b)$ we have
\[
\calM, w_1 \nVdash \exists x R(x) \to \forall x R(x).
\]
Hence, $\calM$ is not a model of the following instance of $\phi$:
\[
\calM, w_1 \nVdash ((P \to Q) \vee (Q \to P)) \vee (\exists x R(x) \to \forall x R(x)).
\]
Now, we want to show that $\mathcal{M}$ is a model of $\mathsf{QFS}$. 

\textbf{Claim.} Take the model $\mathcal{M}$ as above. For any formula $\psi(x)$
\[
\text{either} \quad \calM, w_2 \Vdash \forall x \psi(x)
\quad \text{or} \quad  \calM, w_2 \Vdash \forall x \neg \psi(x)
\]
and
\[
\text{either} \quad \calM, w_3 \Vdash \forall x \psi(x)
\quad \text{or} \quad  \calM, w_3 \Vdash \forall x \neg \psi(x).
\] 
In other words, in any of the final worlds $w_2$ or $w_3$, either both $\psi(a)$ and $\psi(b)$ are valid, or both of them are invalid.

Suppose the claim holds. We prove that $\calM, w_1 \forc \QFS$, because clearly $\calM, w_2 \forc \QFS$ and $\calM, w_3 \forc \QFS$.

\item[$\bullet$] As $F$ is constant domain, $\calM, w_1 \Vdash \CD$.

\item[$\bullet$] To show $\calM, w_1 \Vdash \ED$, for the sake of contradiction assume
\[
\calM, w_1 \Vdash B \to \exists x A(x) \qquad \calM, w_1 \nVdash \exists x (B \to A(x)) \tag{$*$}
\]
for formulas $A(x)$ and $B$ where $x$ is not free in $B$. Therefore,
\[
\calM, w_1 \nVdash B \to A(a)  \qquad \calM, w_1 \nVdash B \to A(b). 
\]
Thus, $\exists i, j \in \{1,2,3\}$ such that
\[
\calM, w_i \Vdash B  \quad \calM, w_i \nVdash A(a) \quad \calM, w_j \Vdash B  \quad \calM, w_j \nVdash A(b). 
\]
The following hold for $i$ and $j$:
\begin{itemize}
\item[$\blacktriangleright$]
$i \neq j$: because if $i=j$, then $\calM, w_i \Vdash B$, $\calM, w_i \nVdash A(a)$, and $\calM, w_i \nVdash A(b)$. As $w_i \succcurlyeq w_1$ this means $\calM, w_1 \nVdash B \to \exists x A(x)$, a contradiction with ($*$).

\item[$\blacktriangleright$]
$i \neq 1$: because if $i=1$, then $\calM, w_1 \Vdash B$ and $\calM, w_1 \nVdash A(a)$. By ($*$), $\calM, w_1 \Vdash A(b)$. However, as $w_j \succcurlyeq w_1$ this means $\calM, w_j \Vdash A(b)$, a contradiction.

\item[$\blacktriangleright$]
$j \neq 1$: similar as the above case.
\end{itemize}
Therefore, w.l.o.g. we can assume $i=2$ and $j=3$. Thus,
\[
\calM, w_2 \Vdash B  \quad \calM, w_2 \nVdash A(a).
\]
However, as $\calM, w_1 \Vdash B \to \exists x A(x)$ we have $\calM, w_2 \Vdash A(b)$, which is a contradiction with the claim.

\item[$\bullet$]  To show $\calM, w_1 \Vdash \SW$, for the sake of contradiction assume
\[
\calM, w_1 \Vdash \forall x A(x) \to B \qquad \calM, w_1 \nVdash \exists x (A(x) \to B) \tag{$\dagger$}
\]
for formulas $A(x)$ and $B$ where $x$ is not free in $B$. Therefore,
\[
\calM, w_1 \nVdash A(a) \to B  \qquad \calM, w_1 \nVdash A(b) \to B. 
\]
Thus, $\exists i, j \in \{1,2,3\}$ such that
\[
\calM, w_i \Vdash A(a)  \quad \calM, w_i \nVdash B \quad \calM, w_j \Vdash A(b)  \quad \calM, w_j \nVdash B. 
\]
The following hold for $i$ and $j$:
\begin{itemize}
\item[$\blacktriangleright$]
$i \neq j$: because if $i=j$, then $\calM, w_i \Vdash \forall x A(x)$. Hence, by $(\dagger)$ we get $\calM, w_i \Vdash B$, a contradiction.

\item[$\blacktriangleright$]
$i \neq 1$: because if $i=1$, then $\calM, w_j \Vdash \forall x A(x)$. Hence, by $(\dagger)$ we get $\calM, w_j \Vdash B$, a contradiction.

\item[$\blacktriangleright$]
$j \neq 1$: similar as the above case.
\end{itemize}
Therefore, w.l.o.g. we can assume $i=2$ and $j=3$. Thus,
\[
\calM, w_2 \Vdash A(a)  \quad \calM, w_2 \nVdash B.
\]
If $\calM, w_2 \Vdash A(b)$, it means that $\calM, w_2 \Vdash \forall x A(x)$ and hence by $(\dagger)$ we have $\calM, w_2 \Vdash B$, which is not possible. Therefore, $\calM, w_2 \nVdash A(b)$, which is a contradiction with the claim.

These two points together prove that $\mathsf{QFS}$ is frame-incomplete. The proofs for the fragments are similar.

\textbf{Proof of the claim:} By an easy induction on the structure of $\psi(x)$. Suppose $\psi(x)$ is an atomic formula. If $x$ is not free in $\psi(x)$, then it is either $P$ or $Q$ and the claim holds trivially. If $\psi(x)=R(x)$, then by definition
\[
\calM, w_2 \Vdash \forall x R(x) \qquad \calM, w_3 \Vdash \forall x R(x).
\]
Now, assume the claim holds for formulas $\alpha(x)$ and $\beta(x)$. We investigate the case of $w_2$. The case of $w_3$ is similar. Note that $w_2$ is a final world and $D_{w_2}=\{a,b\}$. There are four cases:
\[
\calM, w_2 \Vdash \forall x \alpha(x) \qquad \calM, w_2 \Vdash \forall x \beta(x) \tag{1}
\]
\[
\calM, w_2 \Vdash \forall x \alpha(x) \qquad \calM, w_2 \Vdash \forall x \neg \beta(x) \tag{2}
\]
\[
\calM, w_2 \Vdash \forall x \neg \alpha(x) \qquad \calM, w_2 \Vdash \forall x \beta(x) \tag{3}
\]
\[
\calM, w_2 \Vdash \forall x \neg \alpha(x) \qquad \calM, w_2 \Vdash \forall x \neg \beta(x) \tag{4}
\]
$\bullet$ The case $\alpha(x) \wedge \beta(x)$. It is easy to see that
\begin{align*}
\calM, w_2 \Vdash \forall x (\alpha(x) \wedge \beta(x)) \tag{in case of (1)}\\
\calM, w_2 \Vdash \forall x \neg (\alpha(x) \wedge \beta(x)) \tag{in case of (2), (3), (4)}
\end{align*}

$\bullet$ The case $\alpha(x) \vee \beta(x)$. It is easy to see that
\begin{align*}
\calM, w_2 \Vdash \forall x (\alpha(x) \vee \beta(x)) \tag{in case of (1), (2), (3)}\\
\calM, w_2 \Vdash \forall x \neg (\alpha(x) \vee \beta(x)) \tag{in case of (4)}
\end{align*}

$\bullet$ The case $\alpha(x) \to \beta(x)$. It is easy to see that
\begin{align*}
\calM, w_2 \Vdash \forall x (\alpha(x) \to \beta(x)) \tag{in case of (1), (3), (4)}\\
\calM, w_2 \Vdash \forall x \neg (\alpha(x) \to \beta(x)) \tag{in case of (2)}
\end{align*}

$\bullet$ 
For the case of the quantifiers, consider a variable $y$ that may or may not appear in $\alpha(x)$ and may or may not be equal to $x$. We have to show that for $Q \in \{\E,\A\}$: 
\begin{equation}
 \label{qfcase}
 \mathcal{M}, w_2 \Vdash \A x Q y \alpha(x) \text{ or } \mathcal{M}, w_2 \Vdash \A x \neg Q y \alpha(x). \tag{$\dagger$}
\end{equation}
In the case that $\mathcal{M}, w_2 \Vdash \A x \alpha(x)$, it follows by definition that $\mathcal{M}, w_2 \Vdash \A y\A x \alpha(x)$, which implies the left part of \eqref{qfcase}, both in the case that $Q=\A$ and that $Q=\E$. In the case that $\mathcal{M}, w_2 \Vdash \A x \neg\alpha(x)$, it follows that $\mathcal{M}, w_2 \Vdash \A y\A x \neg\alpha(x)$, which implies the right part of \eqref{qfcase}, both in the case that $Q=\A$ and that $Q=\E$.
\end{proof}

\LemPrenex*

\begin{proof}
The {\it nesting depth} of a quantifier is defined inductively as follows. The nesting depth of a formula without quantifiers is zero. If the nesting depth of formulas $C$ and $D$ is $m$ and $n$, respectively, then for $\circ \in \{\en,\of, \imp\}$ a quantifier occurrence in $C$ has nesting depth $n+1$ in $C \circ D$ and nesting depth $n$ in $\E y C(y)$ and $\A y C$. Thus, one counts how ``deep'' a quantifier occurs in terms of connectives, but not in terms of quantifiers. 
For a given formula $A$, let the {\it nesting depth of quantifiers} of $A$, denoted $ndq(A)$, be the sum of the nesting depths of the quantifier occurrences in $A$. 

We prove the lemma with induction to $ndq(A)$ and a subinduction to the complexity of $A$.  In case $ndq(A)=0$, then $A$ does not contain nested quantifies and thus is in prenex normal form. Consider the case that $ndq(A)>0$. Thus $A$ cannot be an atomic formula and therefore has to be of the form $C \circ D$, for some $\circ \in \{\en,\of, \imp\}$, or $Q x C(x)$ for some $Q \in \{\A,\E\}$. 

We start with the first case. There has to be at least one quantifier in $C$ or $D$. Note that $C$ is less complex than $A$ and that $ndq(C)\leq ndq(A)$, and thus, by the induction hypothesis there is a $C'$ that is equivalent to $C$ and in prenex normal form. Since $ndq(C')=0$, $ndq(C'\en D)\leq ndq(C \en D)$. By renaming variables we can assume that $C'=\A xE(x)$ or $C'=\E xE(x)$ and $x$ is not free in $D$. Let $\circ \in \{\en,\of, \imp\}$. 

In case  $C'=\A xE(x)$, for $\circ \in \{\en, \of\}$ we have $L \vdash A \ifff \A x(E(x) \circ D)$, where we use $\CD$ in the case that $\circ=\of$. Note that $ndq(\A x(E(x) \circ D)) < ndq(C'\en D)$ and we can apply the induction hypothesis. In case $\circ =\imp$, we have $L \vdash A \ifff \E x(E(x) \imp D)$ by using $\SW$, and since $ndq(\E x(E(x) \imp D)) < ndq(C'\imp D)$, we can apply the induction hypothesis.

In case  $C'=\E xE(x)$, for $\circ \in \{\en, \of\}$ we have $L \vdash A \ifff \E x(E(x) \circ D)$, where we do not have to use principles from $\QFS$, but just from $\IQC$. In case $\circ =\imp$, we have $L \vdash A \ifff \E x(E(x) \imp D)$ by using $\ED$, and since $ndq(\E x(E(x) \imp D))< ndq(C'\imp D)$, we can apply the induction hypothesis.

We turn to the case that $A = Q x C(x)$ for some $Q \in \{\A,\E\}$. Note that $ndq(C) =  ndq(A)$, but $C$ is less complex then $A$. Therefore, by the subinuction hypothesis we can assume that $D$ is in prenex normal form. But then so is $A$. 
\end{proof}

\ThmFragments*

\begin{proof}
Use similar models as in the proof of Lemma \ref{lem: WF}. The proof is similar to the proof of Lemma \ref{lem: WF validates ED}.

\item[1)]  We have to provide a model $\calM$ such that \[
\calM \vDash \IQC+\{\mathsf{CD}, \mathsf{SW}\} \qquad \text{and} \qquad \calM \not \vDash \ED.
\]    
Take model $\calM =(W , \preccurlyeq, D)$ where $W$ consists of an infinite descending chain $v_0 \succcurlyeq v_1 \succcurlyeq v_2 \succcurlyeq \dots$, called $C$, and a world $w$ such that $w \preccurlyeq v_i$ for each $v_i \in C$. Let $\calM$ be a constant domain model, where the domain is countably infinite, $\calD=\{a_0, a_1, \dots\}$. Then, construct $\calM$ as:
\begin{align*}
\calM , v_i & \Vdash Q & i \geq 1, v_i \in C \\
\forall d \in \calD \ \ \calM , v_0 & \Vdash P(d) \\ 
\forall d \in \calD \ \ \calM , w & \nVdash P(d)   \\ 
\forall d \in \calD - \{a_0, \dots, a_{i-1}\} \ \ \calM, v_i & \Vdash P(d) & i \geq 1, v_i \in C
\end{align*}
depicted as:
\[
  \xymatrix{ 
            \hspace{-15pt}   \forall d \in \mathcal{D} \quad   v_0 \forc Q, P(d)  \\
           \hspace{-45pt}   \forall d \in \mathcal{D} - \{a_0\} \quad   v_1 \forc Q, P(d)  \ar[u]    \\
         \hspace{-58pt}   \forall d \in \mathcal{D} - \{a_0, a_1\} \quad   v_2 \forc Q, P(d)  \ar[u]    \\
           \vdots \ar[u] \\
           w \ar[u]
           }
\]
Similar to the proof of the first item in Lemma \ref{lem: WF}, the following instance of $\ED$ is not valid in $\calM$:
\[
\calM, w \nVdash (Q \to \exists x P(x)) \to \exists x (Q \to P(x)).
\]
As $\calM$ is constant domain, it is a model of $\CD$. Moreover, similar to the proof of the second item in Lemma \ref{lem: WF validates ED}, $\calM$ is a model of $\SW$.

\item[2)]  We have to provide a model $\calM$ such that \[
\calM \vDash \IQC+\{\mathsf{CD}, \mathsf{ED}\} \qquad \text{and} \qquad \calM \not \vDash \SW.
\]    
Take model $\calM =(W , \preccurlyeq, D)$ where $W$ only consists of an infinite descending chain $v_0 \succcurlyeq v_1 \succcurlyeq v_2 \succcurlyeq \dots$. Let $\calM$ be a constant domain model, where the domain is countably infinite, $\calD=\{a_0, a_1, \dots\}$. Then, construct $\calM$ as:
\begin{align*}
\calM , u & \nVdash Q & \forall u \in W \\
\forall d \in \{a_0, \dots, a_{i}\} \ \ \calM, v_i & \Vdash P(d) & \forall i \geq 0
\end{align*}
depicted as:
\[
  \xymatrix{ 
             \vdots \\
           \hspace{45pt}  v_1 \forc P(a_0), P(a_1) \ar[u] \\ 
      \hspace{15pt}   v_0 \forc P(a_0) \ar[u]
           }
\]
Therefore, $\calM, w \nVdash \forall x P(x)$. Thus,
\[
\calM, w \Vdash \forall x P(x) \to Q.
\]
However, as $\calM, v_i \forc P(a_i)$  for every $a_i \in \calD$, we get $\calM, w \nVdash \exists x (P(x) \to Q)$. Hence,
\[
\calM, w \nVdash (\forall x P(x) \to Q) \to \exists x (P(x) \to Q).
\]
On the other hand, as the underlying frame is constant domain, linear and satisfies $\WF$, then $\calM$ is a model for $\ED$.

\item[3)] We have to show that for any model $\mathcal{M}$ 
\[
\text{if} \quad \mathcal{M} \vDash \mathsf{IQC}+\SW \qquad \text{then} \qquad \mathcal{M} \vDash \CD.
\] 
We proceed with the proof by contraposition, i.e., for any model $\mathcal{M}$ of $\IQC$ 
\[
\text{if} \quad \mathcal{M} \not \vDash \CD \qquad \text{then} \qquad \mathcal{M} \not \vDash \SW,
\] 
namely we find an instance of the axiom $\SW$ such that it is not valid in $\mathcal{M}$. Since $\mathcal{M} \not \vDash \CD$, there is a world $u$ and an instance of the axiom $\CD$ such as
\[
B:= \forall x (R(x)\vee P) \to \forall x R(x) \vee P
\]
such that 
\[
\calM, u \nVdash B,
\]
where $P$ is a nullary predicate and $R(x)$ a unary one. As $\calM, u \nVdash B$, there exists a world $w \succcurlyeq u$ such that
\[
\calM, w \Vdash \forall x (R(x)\vee P)  \qquad \calM, w \nVdash \forall x R(x) \vee P.
\]
The right statement implies
\[
\calM, w \nVdash P \qquad \calM, w \nVdash \forall x R(x).
\]
Therefore, we have 
\begin{align*}
\calM, w & \Vdash \forall x (R(x)\vee P) \tag{i} \\
\calM, w & \nVdash P \tag{ii} \\
 \exists w' \succcurlyeq w \; \exists a \in D_{w'} \; \calM, w' & \nVdash R(a).\tag{iii}
\end{align*}
Thus,
\begin{align*}
\calM, w & \nVdash B \tag{by the choice of $w$}\\
\forall d \in D_w \ \calM, w & \Vdash R(d) \tag{by (i),(ii)} \\
a & \notin D_{w} \tag{by (iii) and the previous line}\\
\calM, w' & \Vdash P \tag{by (i) and $w' \succcurlyeq w$}
\end{align*}
This part of the model $\mathcal{M}$ is depicted as:
\[
  \xymatrix{
           \hspace{-10pt}  \ \ w' \Vdash P \\
        \vdots \ar[u]  \\
           \hspace{-25pt}   \forall d \in D_w \ w \Vdash R(d) \ar[u]
           }
\]
Now, take the following instance of the axiom $\SW$:
\[
A: = \big(\forall x R(x) \to B \big) \to \exists x (R(x) \to B).
\]
Note that the variable $x$ is not free in $B$. Now, we want to show that $\mathcal{M} \not \vDash A$. More precisely, we will prove
\[
\calM, w \nVdash \big(\forall x R(x) \to B \big) \to \exists x (R(x) \to B)
\]
by showing that
\begin{align*}
\calM, w & \Vdash  \forall x R(x) \to B \tag{*} \\
\calM, w & \nVdash  \exists x (R(x) \to B) \tag{\dag}
\end{align*}

To prove (*), for any $v \succcurlyeq w$ if we have $\calM, v \Vdash \forall x R(x)$, then $\calM, v \Vdash \forall x R(x) \vee P$. This means that for any $v \succcurlyeq w$ the succedent of $B$ is valid in $v$ which means $\calM, v \Vdash B$. Hence, $\calM, w  \Vdash  \forall x R(x) \to B$.

To prove $(\dag)$, suppose otherwise, i.e., $\calM, w \Vdash  \exists x (R(x) \to B)$. By Definition \ref{Def: Kripke model}, 
\[
\exists d \in D_w \ \forall v \succcurlyeq w \quad \text{if} \quad \calM, v \Vdash R(d) \quad \text{then} \quad \calM, v \Vdash B. 
\]
However, $w \succcurlyeq w$ and $\forall d \in D_w \; \calM, w \Vdash R(d)$ and $\calM, w \nVdash B$. 
Hence, $(\dag)$ holds.

\item[4)] We provide a model\footnote{The credit of the proof goes to an unpublished note by Vítězslav Švejdar.} $\calM$ such that 
\[
\calM \vDash \iqc +\ED \qquad \text{but} \qquad \calM \nvDash \CD.
\]
Let the language only consist of two unary predicates $P$ and $Q$.
Take model $\calM=(W , \preccurlyeq, D)$, where $W$ consists of an infinite ascending chain $w_0 \preccurlyeq w_1 \preccurlyeq \dots$, called $C$, and a final world $\delta$ such that $\delta \succcurlyeq w_i$ for each $w_i \in C$. Denote the domain of each $w_i$ by $D_i=\{0, 1, \dots, i\}$ for any $i \geq 0$ and $D_\delta=\mathbb{N}$, where $\mathbb{N}$ is the set of natural numbers.
\begin{align*}
\calM , w_0 & \Vdash P(0) & \\
\calM , w_0 & \nVdash Q(0) & \\
\forall d \in D_{i-1} \quad \calM , w_i & \Vdash P(d) & i \geq 1, w_i \in C \\
\forall d \in D_{i} \quad \calM , w_i & \Vdash Q(d) & i \geq 1, w_i \in C \\
\forall d \in D_{\delta} \quad \calM , \delta & \Vdash P(d), Q(d) & 
\end{align*}
depicted as:
\[
  \xymatrix{ 
 \hspace{40pt} D_\delta= \mathbb{N} \ \delta \forc P(d), Q(d) \ \forall d \in \mathbb{N} \\
            \vdots \ar[u] \\
           \hspace{35pt}  D_2=\{0,1,2\} \ w_2 \forc P(0), P(1), Q(0), Q(1), Q(2) \ar[u] \\ 
         \hspace{7pt}  D_1=\{0,1\} \ w_1 \forc P(0), Q(0), Q(1) \ar[u] \\ 
     \hspace{-30pt} D_0=\{0\} \   w_0 \forc P(0) \ar[u]
           }
\]
$\calM$ is not a model for $\CD$, because
\[
\calM, w_1 \forc \forall y Q(y) \qquad \calM, w_0 \forc \forall x (\forall y Q(y) \vee P(x)),
\]
and 
\[
\calM, w_0 \nvDash \forall y Q(y) \qquad \calM, w_0 \nvDash \forall x P(x). 
\]
Thus,
\[
\calM, w_0 \nvDash \forall x (\forall y Q(y) \vee P(x)) \to (\forall y Q(y) \vee \forall x P(x)). 
\]
Now, we want to show that $\calM$ is a model for $\ED$, i.e., $\forall w \in W$
\[
\calM, w \forc B \to \exists x A(x) \quad \text{then} \quad \calM, w \forc \exists x (B \to A(x)) \tag{$*$}
\]
for formulas $A$ and $B$, where $x$ is not free in $B$. The following claim is easy to prove:

\textbf{Claim 1. (homogeneity of $\delta$)}  For any formula $\varphi$, if $\calM, \delta \forc \varphi[e]$ for some valuation $e$, then $\delta \forc \varphi[e]$ for all valuations $e$. 

Now, define a \emph{setup} as a pair $\left[w_n, e\right]$, where $n \neq 0$ and $e$ is a valuation of variables in $D_{w_n}$. We say two setups $\left[w_n, e_1\right]$ and $\left[w_m, e_2\right]$ are \emph{similar} when $e_1(y)=n$ if{f} $e_2(y)=m$, for each variable $y$. 

\textbf{Claim 2.} For any formula $\phi$, if $\left[w_n, e_1\right]$ and $\left[w_m, e_2\right]$ are similar setups, then 
\[
w_n  \forc \varphi\left[e_1\right] \qquad \text{if{f}} \qquad w_m \forc \varphi\left[e_2\right].
\]
Now, to prove that $\calM$ is a model of $\ED$, assume $\calM, w_0 \forc B \to \exists x A(x)$, where $x$ is not free in $B$. Then, if $B$ is valid nowhere, clearly $\calM, w_0 \forc B \to A(0)$, hence $\calM, w_0 \forc \exists x (B \to A(x))$. Therefore, assume there exists a world such that $B$ is valid in it. Take $v=$ min $\{w \in W \mid \calM, w \forc B\}$. There are three cases:

\item[$\bullet$] $v=w_0$. Then, clearly $\calM, v \forc A(0)$.

\item[$\bullet$] $v = \delta$. Then, $\calM, \delta \forc \exists x A(x)$. By claim 1, $\calM, \delta \forc A(0)$.

\item[$\bullet$] $v=w_i$ for an $i \geq 1$. Thus,
\[
\exists d \in \{0, 1, \dots, i\} \qquad \text{such that} \qquad \calM, w_i \forc A(d).
\]
If $d \neq i$, then by Claim 2, $\calM, w_i \forc A(0)$. If $d=i$, then $\calM, w_{i+1} \forc A(d)$. Note that $d \neq i+1$. Thus, by Claim 2, $\calM, w_i \forc A(0)$.

In each case, as $v$ is the smallest world such that $B$ is valid in it, and $\calM, v \forc A(0)$, we get $\calM, w_0 \forc B \to A(0)$, i.e.,
\[
\calM, w_0 \forc \exists x (B \to A(x)).
\]
Similarly we can prove $(*)$ holds for any $w \in W$, i.e., $\ED$ is valid in $\calM$.

\textbf{Proof of Claim 2.} We use induction on the complexity of $\varphi$. The only challenging case is when $\varphi$ is $\psi \rightarrow \chi$. If $w_n \forc (\psi \rightarrow \chi)\left[e_1\right]$, then take a world $\beta \succcurlyeq w_m$ such that $\beta \forc \psi\left[e_2\right]$. The cases $\beta=w_m$ and $\beta=\delta$ are straightforward. If $\beta \neq w_m$ and $\beta \neq \delta$, it is sufficient to notice that the setups $\left[w_{n+1}, e_1\right]$ and $\left[\beta, e_2\right]$ are similar. Assume that $\varphi$ is $\forall x \psi$ and $w_n \forc (\forall x \psi)\left[e_1\right]$. Take a node $\beta \succcurlyeq w_m$ and an element $b \in D_{\beta}$. If $\beta=\delta$, then $\beta \forc \psi\left[e_2(x / b)\right]$ by persistency and the homogeneity of $\delta$. If $\beta \neq \delta$, then depending on whether $\beta=w_m$ or $\beta \neq w_m$, it is possible to choose an $a$ such that $\left[\beta, e_2(x / b)\right]$ is similar to $\left[w_n, e_1(x / a)\right]$ or to $\left[w_{n+1}, e_1(x / a)\right]$.

\vspace{10pt}

\noindent \textbf{A syntactic proof for 3)}

We could prove $\IQC + \SW \vdash \CD$ syntactically as follows. Using
\[
\LJ \vdash A(t), A(t) \to \forall x A(x) \vee B \Rightarrow \forall x A(x) \vee B
\]
and 
\[
\LJ \vdash B, A(t) \to \forall x A(x) \vee B \Rightarrow \forall x A(x) \vee B
\]
we get
\[
\LJ \vdash A(t) \vee B, A(t) \to \forall x A(x) \vee B \Rightarrow \forall x A(x) \vee B
\]
and we have in $\LJ$:
\begin{center} 
\AxiomC{$A(t) \vee B, A(t) \to \forall x A(x) \vee B \Rightarrow \forall x A(x) \vee B$}
 \UnaryInfC{$\forall x (A(x) \vee B), A(t) \to \forall x A(x) \vee B \Rightarrow \forall x A(x) \vee B$}
 \UnaryInfC{$\forall x (A(x) \vee B), \exists x (A(x) \to \forall x A(x) \vee B) \Rightarrow \forall x A(x) \vee B$}
 \DisplayProof
\end{center}
\normalsize Moreover, using the fact that the sequents
\[
\Rightarrow (\forall x A(x) \to \forall x A(x) \vee B) \to \exists x (A(x) \to \forall x A(x) \vee B)\]
and
\[
\Rightarrow \forall x A(x) \to \forall x A(x) \vee B
\]
are provable in $\mathbf{LJ}+\SW$, we get 
\[
\mathbf{LJ}+\SW \vdash  \Rightarrow \exists x (A(x) \to \forall x A(x) \vee B)
\]
\normalsize where the first sequent is an instance of the axiom $\SW$.
Therefore, using the cut rule on the above sequent and the conclusion of the proof tree, we get the axiom $\CD$
\[
\forall x (A(x) \vee B) \Rightarrow \forall x A(x) \vee B.
\]
\end{proof}

\end{document}